\documentclass[sn-mathphys-num]{sn-jnl}
\usepackage{graphicx}%
\usepackage{multirow}%
\usepackage{amsmath,amssymb,amsfonts}%
\usepackage{amsthm}%
\usepackage{mathrsfs}%
\usepackage[title]{appendix}%
\usepackage[dvipsnames]{xcolor}%
\usepackage{textcomp}%
\usepackage{manyfoot}%
\usepackage{booktabs}%
\usepackage{algorithm}%
\usepackage{algorithmicx}%
\usepackage{algpseudocode}%
\usepackage{listings}%
\usepackage[normalem]{ulem}

\usepackage{tikz}
\usetikzlibrary{automata, arrows.meta, positioning, shapes}
\usetikzlibrary{arrows.meta, 
                automata,
                intersections,
                positioning,
                quotes,
                matrix,
                bbox 
                }
\usetikzlibrary{shapes,arrows,fit}
\newtheorem{theorem}{Theorem}
\newtheorem{definition}{Definition}
\newtheorem{proposition}{Proposition}
\newtheorem{lemma}{Lemma}
\newtheorem{corollary}{Corollary}
\newtheorem{example}{Example}

\newcommand{\nf}{\mathrm{nf}}
\newcommand{\newj}{e }

\begin{document}
\title[Positive Characteristic Sets for Relational Pattern Languages]{Positive Characteristic Sets for Relational Pattern Languages}

\author*[1]{\fnm{S.\ Mahmoud} \sur{Mousawi}}\email{mousawi.s.m@gmail.com}

\author[1,2]{\fnm{Sandra} \sur{Zilles}}\email{sandra.zilles@uregina.ca}

\affil*[1]{\orgdiv{Department of Computer Science}, \orgname{University of Regina}, \orgaddress{\country{Canada}}}

\affil[2]{\orgdiv{Alberta Machine Intelligence Institute}}

%
%
%
%
%
\abstract{In the context of learning formal languages, data about an unknown target language $L$ is given in terms of a set of \emph{(word,label)}\/ pairs, where a binary label indicates whether or not the given word belongs to $L$. A (polynomial-size) characteristic set for $L$, with respect to a reference class $\mathcal{L}$ of languages, is a set of such pairs that satisfies certain conditions allowing a learning algorithm to (efficiently) identify $L$ within~$\mathcal{L}$. In this paper, we introduce the notion of \emph{positive characteristic set}, referring to characteristic sets of only positive examples. These are of importance in the context of learning from positive examples only. We study this notion for classes of relational pattern languages, which are of relevance to various applications in string processing.}

\maketitle     
\section{Introduction}

Many applications in machine learning and database systems, such as, e.g., the analysis of protein data~\cite{ArikawaMSKMS93}, the design of algorithms for program synthesis~\cite{Nix85}, or problems in pattern matching~\cite{CliffordHPS09}, deal with the problem of finding and describing patterns in sets of strings. From a formal language point of view, one way to address this problem is with the study for so-called pattern languages~\cite{Angluin80patterns,Shinohara82}.

A pattern is a finite string of variables and terminal symbols. For instance, if the alphabet of terminal symbols is $\Sigma=\{a,b\}$, and $x_1$, $x_2$, \dots, denote variables, then $p=ax_1 ba x_1x_2$ is a pattern. The language generated by this pattern consists of all words that are obtained when replacing variables with finite words over~$\Sigma$, where multiple occurrences of the same variable are replaced by the same word. Angluin~\cite{Angluin80patterns} further required that no variable be replaced by the empty string, resulting in the notion of \emph{non-erasing}\/ pattern language, while Shinohara's \emph{erasing}\/ pattern languages allow erasing variables. For the pattern $p$ above, the non-erasing language is $\{aw_1 ba w_1w_2\mid w_1,w_2\in\Sigma^+\}$, while the erasing language is $\{aw_1 ba w_1w_2\mid w_1,w_2\in\Sigma^*\}$. The latter obviously contains the former.

The repetition of the variable $x_1$ in $p$ sets the words replaced for the two occurrences of $x_1$ in \emph{equality}\/ relation. Angluin already noted that relations other than equality could be studied as well--- a thought that was later on explored by Geilke and Zilles~\cite{GeilkeZ11}. For example, the \emph{reversal}\/ relation requires that the words replaced for the two occurrences of $x_1$ are the reverse of one another, which can be useful for modeling protein folding or behaviour of genes~\cite{GeilkeZ11}. To the best of our knowledge, the first work to focus on specific relations aside from equality is that of Holte et al.~\cite{HolteMZ22}, which studies the relations \emph{reversal}\/ as well as \emph{equal length}, in the non-erasing setting. As the name suggests, \emph{equal length}\/ stipulates that the words replaced for the two occurrences of $x_1$ in $p$ are of the same length. 

Holte et al.\ focused on the question whether the equivalence of two patterns can be decided efficiently. While deciding equivalence is quite simple for non-erasing pattern languages with the equality relation~\cite{Angluin80patterns}, the situation is more complex for reversal and equal length relations. Holte et al.'s proposed decision procedure, which tests only a small number of words for membership in the two pattern languages under consideration, has applications beyond deciding equivalence. In particular, the sets of tested words serve as so-called teaching sets in the context of machine teaching~\cite{ZhuSZR18} as well as characteristic sets in the context of grammatical inference from positive and negative data~\cite{Higuera97}.

We continue this line of research in two ways: firstly, we study reversal and equal length relations for \emph{erasing}\/ pattern languages; secondly, we focus on the concept of characteristic set rather than effectively deciding the equivalence problem. From that perspective, there is a direct connection of our work to the theory of learning relational pattern languages. Our contributions are as follows:

(a) We introduce the notion of \emph{positive characteristic sets}. Classical characteristic sets contain positive and negative examples, i.e., words over $\Sigma$ together with a label indicating whether or not the word belongs to a target language to be learned. While characteristic sets are crucial for studying the (efficient) learnability of formal languages from positive and negative data~\cite{Higuera97}, they are not suited to the study of learning from positive examples only. Positive characteristic sets are an adjustment to account for learning from only positive examples. 

(b) Perhaps not surprisingly, we show that the existence of positive characteristic sets is equivalent to the existence of telltale sets, which are known to characterize classes of languages learnable from positive data in Gold's~\cite{Gold67} model of learning in the limit~\cite{Angluin80telltales}. Moreover, \emph{small}\/ positive characteristic sets yield \emph{small}\/ telltales. We thus make connections between grammatical inference from characteristic sets and learning in the limit using telltales.

(c) Holte et al.\ conjectured that the class of non-erasing pattern languages under the reversal relation possesses a family of small sets which, in our terminology, would be linear-size positive characteristic sets. We prove that, in the erasing case under the reversal relation, \emph{no}\/ family of positive characteristic sets exists (let alone a family of small such sets) when $|\Sigma|=2$. 

(d) For equal length, Holte et al.\ showed that a family of linear-size telltales exists for the class of non-erasing pattern languages, assuming $|\Sigma|\ge 3$. We first adapt their proof idea to show that this statement remains true in the erasing case. However, the same proof approach cannot be translated to binary alphabets. Our main result states that even if $|\Sigma|=2$, a non-trivial subclass of the class of erasing pattern languages under the equal length relation has linear-size telltales (and thus has positive characteristic sets). The proof techniques required for obtaining this result are vastly different from those used for the case $|\Sigma|\ge 3$. The subclass for which our result holds consists of all erasing pattern languages (under the equal length relation) for which the underlying patterns have all their terminal symbols occurring in blocks of at least three consecutive symbols (except maximal terminal blocks of the shape $a^nb$ or $ab^n$ where $\Sigma=\{a,b\}$) and satisfy an additional structural constraint on the  relations between variables. 

An earlier conference version of this paper claimed a stronger result, without limiting the class of patterns, see~\cite[Theorem 2]{MousawiZ24}. However, the underlying proof was flawed; we here present a major change in the original argument, while establishing a weaker claim than the one made in~\cite{MousawiZ24}.

\section{Notation and Preliminary Results}

Throughout this paper, for any finite set $M$ and any $n\in\mathbb{N}$, we use $M^{\ge n}$ to denote the set of all finite strings of length at least $n$ over $M$. Moreover, following standard notation, $M^+:=M^{\ge 1}$ denotes the set of all finite non-empty strings over $M$, while $M^*:=M^{\ge 0}$ refers to the set of all finite strings over $M$, including the empty string.

We use $\Sigma$ to denote an arbitrary finite alphabet. A language over $\Sigma$ is simply a subset of $\Sigma^*$, and every string in $\Sigma^*$ is called a word. We denote the empty string/word by $\varepsilon$. The reverse operation on a word $w=w_1\ldots w_l\in\Sigma^*$, with $l\in\mathbb{N}$ and $w_i\in\Sigma$, is defined by $w^{rev} := w_l \ldots w_1$. For any set $Z$ and any string $s\in Z^*$, we use $|s|$ to refer to the length of $s$, $s[1]$ to refer to the first symbol in $s$, and $s[-1]$ to refer to the last symbol in $s$. For $1\le i,j\le|s|$, we use $s[i:j]$ to denote the substring of length $j-i+1$ of $s$, starting at the $i$th and ending at the $j$th position of~$s$. 

Let $L$ be a language over $\Sigma$. A set $S\subseteq\Sigma^*\times\{0,1\}$ is consistent with $L$ iff ($(s,1)\in S$ implies $s\in L$) and ($(s,0)\in S$ implies $s\notin L$). When learning a language $L$ from positive and negative examples (without noise), a learning algorithm is presented with a set $S\subseteq\Sigma^*\times\{0,1\}$ that is consistent with $L$; from this set it is supposed to identify $L$. We refer the reader to~\cite{Higuera97} for details of learning formal languages; the present paper will focus on structural properties of classes of formal languages that allow for efficient learning of formal languages from \emph{only positive}\/ examples; in particular, we consider sets $S\subseteq\Sigma^*\times\{1\}$. 

\subsection{Positive Characteristic Sets}

While learning languages from positive examples is well-studied~\cite{Angluin80telltales,Gold67,LangeZZ08}, we build a bridge between  learnability from positive examples (via so-called \emph{telltales}~\cite{Angluin80telltales}) and (efficient) learnability from positive and negative examples (via (polynomial-size) \emph{characteristic sets}~\cite{Higuera97}).

\medskip

\begin{definition}[adapted from~\cite{Higuera97}]
Let $\mathcal{L}=(L_i)_{i\in\mathbb{N}}$ be a family of languages over $\Sigma$. A family $(C_i)_{i\in\mathbb{N}}$ of finite subsets of $\Sigma^*\times\{0,1\}$ is called a family of characteristic sets for $\mathcal{L}$~iff: 
\begin{enumerate}
    \item $C_i$ is consistent with $L_i$ for all~$i\in\mathbb{N}$.
    \item If $L_i\ne L_j$, and $C_i$ is consistent with $L_j$, then $C_j$ is not consistent with~$L_i$.
\end{enumerate}
If, in addition, $C_i\subseteq\Sigma^*\times\{1\}$ for all $i\in\mathbb{N}$, then $(C_i)_{i\in\mathbb{N}}$ is called a family of positive characteristic sets for~$\mathcal{L}$. 
\end{definition}

\medskip

In the context of machine teaching, characteristic sets have also been called \emph{non-clashing teaching sets}~\cite{FallatKSSZ23}. 

Every family $\mathcal{L}=(L_i)_{i\in\mathbb{N}}$ possesses a family of characteristic sets; it suffices to include in $C_i$, for each $j<i$ with $L_j\ne L_i$, one example from the symmetric difference of $L_j$ and $L_i$, labelled according to $L_i$. In the literature, the \emph{size}\/ and \emph{effective computability}\/ of characteristic sets are therefore the main aspects of interest, and relate to notions of efficient learning of languages~\cite{Higuera97}. However, the mere \emph{existence}\/ of a family of positive characteristic sets is a non-trivial property. We will show below that it is equivalent to the existence of so-called telltales.

\medskip

\begin{definition}[\cite{Angluin80telltales}]
Let $\mathcal{L}=(L_i)_{i\in\mathbb{N}}$ be a family of languages over $\Sigma$. A family $(T_i)_{i\in\mathbb{N}}$ of finite subsets of $\Sigma^*$ is called a family of telltales for $\mathcal{L}$~iff: 
\begin{enumerate}
    \item $T_i\subseteq L_i$ for all $i\in\mathbb{N}$.
    \item If $i,j\in\mathbb{N}$, $L_i\ne L_j$, and $T_i\subseteq L_j$, then $L_j\not\subset L_i$.
\end{enumerate}
\end{definition}

For families  $\mathcal{L}=(L_i)_{i\in\mathbb{N}}$ with uniformly decidable membership, the existence of a family of telltales is equivalent to the learnability of $\mathcal{L}$ from positive examples only, in Gold's~\cite{Gold67} model of learning in the limit~\cite{Angluin80telltales}. We connect telltales to characteristic sets by observing that the existence of a telltale family is equivalent to the existence of a family of positive characteristic sets.

\medskip

\begin{proposition}\label{prop:telltales}
Let $\mathcal{L}=(L_i)_{i\in\mathbb{N}}$ be a family of languages over $\Sigma$. Then the following two statements are equivalent.
\begin{enumerate}
    \item $\mathcal{L}$ possesses a family of positive characteristic sets.
    \item $\mathcal{L}$ possesses a family of telltales.
\end{enumerate}
\end{proposition}

\begin{proof}
To show that Statement 1 implies Statement 2, let $(C_i)_{i\in\mathbb{N}}$ be a family of positive characteristic sets for $\mathcal{L}$. Define $C'_i=\{s\mid (s,1)\in C_i\}$ for all $i\in\mathbb{N}$. We argue that $(C'_i)_{i\in\mathbb{N}}$ is a family of telltales for $\mathcal{L}$.
To see this, suppose there are $i,j\in\mathbb{N}$ such that $L_i\ne L_j$, $C'_i\subseteq L_j$, and $L_j\subset L_i$. Then $C'_j\subset L_i$, so that $C_j$ is consistent with $L_i$. Further, $C'_i\subseteq L_j$ implies that $C_i$ is consistent with $L_j$. This is a contradiction to $(C_i)_{i\in\mathbb{N}}$ being a family of positive characteristic sets for $\mathcal{L}$. Hence no such $i,j\in\mathbb{N}$ exist, i.e., $(C'_i)_{i\in\mathbb{N}}$ is a family of telltales for~$\mathcal{L}$.
    
To verify that Statement 2 implies Statement 1, let $(T_i)_{i\in\mathbb{N}}$ be a family of telltales for $\mathcal{L}$. The following (possibly non-effective) procedure describes the construction of a set $C_i$ for any $i\in\mathbb{N}$, starting with $C_i=\emptyset$:
\begin{itemize}
    \item For each $s\in T_i$, add $(s,1)$ to $C_i$.
    \item For each $j<i$, if $L_i\not\subseteq L_j$, pick any $s_j\in L_i\setminus L_j$ and add $(s_j,1)$ to $C_i$.
\end{itemize}
Obviously, $C_i$ is finite. Now suppose there are $i,j$ with $j<i$ and $L_i\ne L_j$, such that $C_i\cup C_j$ is consistent with $L_i$ and $L_j$. If $L_i\not\subseteq L_j$, then, by construction, $C_i$ contains $(s_j,1)$ for some $s_j\notin L_j$, so that $C_i$ is not consistent with $L_j$---a contradiction. So $L_i\subseteq L_j$. As $C_j$ is consistent with $L_i$, we have $T_j\subseteq L_i\subseteq L_j$. By the telltale property, this yields $L_i=L_j$---again a contradiction.
\end{proof}

While the existence of a family of characteristic sets is trivially fulfilled, not every family $\mathcal{L}=(L_i)_{i\in\mathbb{N}}$ has a family of positive characteristic sets. For example, it was shown that the family containing $\Sigma^*$ as well as each finite subset of $\Sigma^*$ does not possess a family of telltales~\cite{Angluin80telltales}. Therefore, by Proposition~\ref{prop:telltales}, it does not possess a family of positive characteristic sets. 

In our study below, the size of positive characteristic sets (if any exist) will be of importance. The term 
\emph{family of polynomial-size positive characteristic sets/telltales}\/ refers to a family of positive characteristic sets/telltales in which the size of the set $C_i$/$T_i$ assigned to language $L_i$ is bounded by a polynomial in the size of the underlying representation of $L_i$. The proof of Proposition~\ref{prop:telltales} now has an additional consequence of interest in the context of efficient learning:

\medskip

\begin{corollary}\label{cor:polysize}
    Let $\mathcal{L}=(L_i)_{i\in\mathbb{N}}$ be a family of languages over $\Sigma$. If $\mathcal{L}$ possesses a family of polynomial-size positive characteristic sets then $\mathcal{L}$ possesses a family of polynomial-size telltales.
\end{corollary}

\medskip

This corollary motivates us to focus our study on positive characteristic sets rather than on telltales---bounds on the size of the obtained positive characteristic sets then immediately carry over to telltales.


\subsection{Relational Pattern Languages}

Let $X=\{x_1,x_2,\ldots\}$ be a countably infinite set of variables, disjoint from $\Sigma$. 
Slightly deviating from Angluin's notation, a finite non-empty string $p$ over $\Sigma\cup X$ is called a pattern, if no element from $X$ occurs multiple times in $p$; the set of variables occurring in $p$ is then denoted by $\mathit{Var}(p)$. For instance, for $\Sigma=\{a,b\}$, the string $p=x_1ax_2abx_3x_4b$ is a pattern with $\mathit{Var}(p)=\{x_1,\ldots,x_4\}$. A pattern $p$ generates words over $\Sigma$ if we substitute each variable with a word, i.e., we apply a substitution $\varphi$ to $p$. A substitution is a mapping $(\Sigma\cup X)^+\rightarrow\Sigma^*$ that acts as a morphism with respect to concatenation and is the identity when restricted to $\Sigma^+$. The set of words generated from a pattern $p$ by applying any substitution is a regular language, unless one constrains the substitutions. In the above example, $p$ generates the regular language $\{w_1aw_2abw_3w_4b\mid w_i\in\Sigma^*\mbox{ for }1\le i\le 4\}$.

If $p$ is any pattern, then a substring $\vec{x}$ of $p$ is called a variable block of $p$ if $\vec{x}\in X^+$ and there are $\omega_1,\omega_2\in\Sigma^+$ such that $p$ starts with $\vec{x} \omega_1$ or ends with $\omega_2\vec{x}$ or contains the substring $\omega_1\vec{x} \omega_2$. Similarly, a substring $\omega$ of $p$ is called a terminal block of $p$ if $\omega\in\Sigma^+$ and there are $\vec{x}_1,\vec{x}_2\in X^+$ such that $p$ starts with $\omega \vec{x}_1$ or ends with $\vec{x}_2\omega$ or contains the substring $\vec{x}_1 \omega \vec{x}_2$. In the above example, the variable blocks of $p$ are $x_1$, $x_2$, and $x_3x_4$, while the terminal blocks of $p$ are $a$, $ab$, and~$b$.

One way of constraining substitutions is to force them to obey relations between variables. Angluin~\cite{Angluin80patterns} allowed any subsets of variables to be in \emph{equality}\/ relation, requiring that all variables in equality relation must be substituted by identical strings. Furthermore, she constrained substitutions to be \emph{non-erasing}, i.e., they map every variable to a non-empty string over $\Sigma$. For example, the pattern $p=x_1ax_2abx_3x_4b$, under the constraint that $x_3,x_4$ are in equality relation and only non-erasing substitutions are allowed, generates the (non-regular!) language $\{w_1aw_2abw_3w_3b\mid w_i\in\Sigma^+\mbox{ for }1\le i\le 3\}$. Later, Shinohara removed the constraint of non-erasing substitutions, resulting in so-called erasing pattern languages~\cite{Shinohara82}. Aside from equality, two relations that were studied in the context of non-erasing pattern languages are \emph{reversal}\/ and \emph{equal length}~\cite{HolteMZ22}. We denote by $eq$, $rev$, and $len$, resp., the relations on $\Sigma^*$ that correspond to equality, reversal, and equal length, resp. If $w, v\in\Sigma^*$, then (i) $(w,v)\in eq$ iff $w=v$; (ii) $(w,v)\in rev$ iff $w=v^{rev}$; (iii) $(w,v)\in len$ iff $|w|=|v|$. For instance, $(ab,bb)\in len\setminus(rev\cup eq)$, $(aab,baa)\in (rev\cup len)\setminus eq$, $(ab,ab)\in (eq\cup len)\setminus rev$, and $(aa,aa)\in eq\cap rev\cap len$. Clearly, $rev\cup eq\subseteq len$.

In general, a relational pattern is a pair $(p,R)$, where $p$ is a pattern and $R$ is a binary relation over $X$. If $r$ is a fixed binary relation over $\Sigma^*$ and $\varphi$ is any substitution, we say that $\varphi$ is valid for $R$ (or $\varphi$ is an $R$-substitution), if $(x,y)\in R$ implies $(\varphi(x),\varphi(y))\in r$. For example, if $r=len$, $\varphi(x)=bb$, $\varphi(y)=ab$, and $\varphi(z)=aaa$, then $\varphi$ would be valid for $R=\{(x,y)\}$, but not for $R'=\{(x,z)\}$ or for $R''=\{(x,y),(x,z)\}$. 

If $(p,R)$ is a relational pattern, then a \emph{group}\/ in $(p,R)$ is any maximal set of variables in $p$ that are pairwise related via the transitive closure of the symmetric closure of $R$ over $\mathit{Var}(p)$. For example, when $p= ax_1x_2bax_3abx_4x_5x_6$ and $R=\{(x_1,x_3),(x_1,x_4),(x_2,x_6)\}$, then the groups are $\{x_1,x_3,x_4\}$, $\{x_2,x_6\}$, $\{x_5\}$. The group of a variable $x$ is the group to which it belongs; we denote this group by $[x]$. Sometimes, we write $[x]_R$ instead of $[x]$ to avoid confusion about the underlying relation $R$. A variable is \emph{free}\/ if its group is of size 1, like $x_5$ in our example.

\section{Reversal Relation}

In this section, we will fix $r=rev$. Note that a substitution that is valid for $\{(x,y),(y,z),(x,z)\}$, where $x,y,z$ are three mutually distinct variables, must replace $x,y,z$ all with the same string, and this string must be a palindrome. For convenience of the formal treatment of the $rev$ relation, we avoid this special situation and consider only \emph{reversal-friendly}\/ relational patterns; a relational pattern $(p,R)$ is reversal-friendly, if the transitive closure of the symmetric closure of $R$ does not contain any subset of the form $\{(y_1,y_2),(y_2,y_3),\ldots,(y_{n},y_1)\}$ for any odd $n\ge 3$ with pairwise distinct variables~$y_i$.

Note that, for a reversal-friendly pattern $(p,R)$, any group $[x]$ is partitioned into two subsets $[x]_{rev}$ and $[x]_=$. The former consists of the variables in $[x]$ for which any $R$-substitution $\varphi$ is forced to substitute $\varphi(x)^{rev}$, while the latter contains variables in $[x]$ which $\varphi$ replaces with $\varphi(x)$ itself.  Clearly, $x\in[x]_=$.

If $(p,R)$ is a reversal-friendly relational pattern, then $\varepsilon L_{rev}(p,R):=\{\varphi(p)\mid \varphi\mbox{ is an \textit{R}-substitution}\}$. The class of (erasing) reversal pattern languages is now given by 
$\varepsilon\mathcal{L}_{rev}=\{\varepsilon L_{rev}(p,R)\mid (p,R)\mbox{ is a reversal-friendly relational pattern}\}$.

In the literature~\cite{HolteMZ22}, so far only the non-erasing version of these languages was discussed: $L_{rev}(p,R):=\{\varphi(p)\mid \varphi\mbox{ is a non-erasing \textit{R}-substitution}\}$. The class of non-erasing reversal pattern languages is then defined as $\mathcal{L}_{rev}=\{L_{rev}(p,R)\mid (p,R)\mbox{ is a reversal-friendly relational pattern}\}$. From~\cite{HolteMZ22}, it is currently still open whether or not \emph{polynomially-sized}\/ positive characteristic sets exist for $\mathcal{L}_{rev}$. Using basic results on telltales, one can easily show that, for any alphabet size, the class $\mathcal{L}_{rev}$ has a family of positive characteristic sets.

\medskip

\begin{proposition}
    Let $\Sigma$ be any countable alphabet. Then $\mathcal{L}_{rev}$ has a family of positive characteristic sets.
\end{proposition}

\begin{proof}
    Let $w$ be any word in $L_{rev}(p,R)$. Since non-erasing substitutions generate words at least as long as the underlying pattern, we have $|p'|\le |w|$ for all $p'$ for which there is some $R'$ with $w\in L_{rev}(p',R')$. Thus, there are only finitely many languages $L\in \mathcal{L}_{rev}$ with $w\in L$. It is known that any family for which membership is uniformly decidable (like $\mathcal{L}_{rev}$), and which has the property that any word $w$ belongs to at most finitely many members of the family, possesses a family of telltales~\cite{Angluin80telltales}. By Proposition~\ref{prop:telltales}, thus $\mathcal{L}_{rev}$ has a family of positive characteristic sets.
\end{proof}

When allowing erasing substitutions, we obtain a contrasting result. In particular, we will show that, when $|\Sigma|=2$, the class $\varepsilon\mathcal{L}_{rev}$ does not have a family of positive characteristic sets.

For convenience of notation, for a relational pattern $(p,R)$, we will write $y=x^{rev}$ and use $x^{rev}$ interchangeably with $y$, when $(x,y)\in R$, where $x,y$ are variables in $p$. Thus, we can rewrite a relational pattern $(p,R)$ in alternate form $\overline{(p,R)}$. For example, we can write $\overline{(p,R)}= x_1 x^{rev}_1 x_2 x^{rev}_2 x_3 x^{rev}_3$ instead of $(p,R)= (x_1 y_1 x_2 y_2 x_3 y_3,\{(x_1,y_1),(x_2,y_2),(x_3,y_3)\})$. The reverse operation on a string of variables is defined as 
    \[(\nu_1 \nu_2 \ldots \nu_{l-1} \nu_l)^{rev} := (\nu_l)^{rev} (\nu_{l-1})^{rev} \ldots (\nu_2)^{rev} (\nu_1)^{rev}\,,\]
where $\nu_i \in \{x_i, x^{rev}_i\}$ for all $1 \leq i \leq l$. Note that $(x^{rev}_i)^{rev}:=x_i$ and $(x_i)^{rev}:=x^{rev}_i$. Now let $\bar{X}=X\cup\{x^{rev}\mid x\in X\}$. 

Let $R$ be any binary relation over $X$. A reversal-obedient morphism for $R$ is a mapping $\phi:(\Sigma\cup\bar{X})^*\rightarrow (\Sigma\cup\bar{X})^*$ that acts as a morphism with respect to concatenation, leaves elements of $\Sigma$ unchanged, and satisfies the following property for all $p\in (\Sigma\cup\bar{X})^*$: if $x,y\in \bar{X}$ occur in $p$ and $(x,y)\in R$ then $\phi(x)=\phi(y)^{rev}$. For example, if $R=\{(x,y)\}$ and $a\in\Sigma$, such a mapping $\phi$ could be defined by $\phi(\sigma)=\sigma$ for all $\sigma\in\Sigma$, $\phi(x)=x_1x_2$, $\phi(y)=x^{rev}_2x^{rev}_1$, and $\phi(z)=ax_3$. This mapping $\phi$ would yield $\phi(axyaz)=\phi(axx^{rev}az)=ax_1x_2x^{rev}_2x^{rev}_1aax_3$.

The following lemma is helpful for establishing that positive characteristic sets do not in general exist for $\varepsilon\mathcal{L}_{rev}$. An analogous version for the equality relation was proven in~\cite{JSSY95}; our proof for the reversal relation is very similar.

\medskip

\begin{lemma}\label{lem-cf-inclusion}
    Let $|\Sigma| \geq 2$. Suppose $(p, R) , (p', R')$ are two arbitrary relational patterns where $p,p'\in X^+$.  Then $\varepsilon L_{rev}(p, R) \subseteq \varepsilon L_{rev}(p', R')$ if and only if there exists a reversal-obedient morphism $\phi$ for $R'$ such that $\phi\overline{(p', R')} = \overline{(p, R)}$.
    
\end{lemma}
\begin{proof} If there exists a reversal-obedient morphism $\phi$ such that $\phi\overline{(p', R')} = \overline{(p, R)}$
    then $\varepsilon L_{rev}(p, R) \subseteq \varepsilon L_{rev}(p', R')$ follows trivially. 
    So suppose $\varepsilon L_{rev}(p, R) \subseteq \varepsilon L_{rev}(p', R')$. 
    
    Let $p = x_1 x_2 \ldots x_{m}$ and $p' = y_1 y_2 \ldots y_{m'}$ where $x_1 x_2 \ldots x_{m}$, $y_1 y_2 \ldots y_{m'} \in X$ and $m, m' \geq 1$. 
    Now, define a substitution $\theta$ as follows, by sequentially defining $\theta(x_i)$ for $i=1,2,\ldots$: If there is some $j<i$ such that $x_i\in[x_j]$, then let $\theta(x_i)=\theta(x_j)^{rev}$, in case $x_i\in[x_j]_{rev}$, and $\theta(x_i)=\theta(x_j)$, in case $x_i\in[x_j]_=$. Otherwise, define
    \[ \theta(x_i) = a b^{(2m' i)+1} a ~ a b^{(2m' i)+2} a \ldots  a b^{(2m' i)+ 2m'} a \]
    Note that for any two variables $x_i$ and $x_j$ where $x_i \notin[x_j]$, there is no common substring of the form $ab^+a$ in $\theta(x_i)$ and $\theta(x_j)$. Now $\theta(p) \in \varepsilon L_{rev}(p, R)$. 

    Since $\varepsilon L_{rev}(p, R) \subseteq \varepsilon L_{rev}(p', R')$ we have $\theta(p) \in \varepsilon L_{rev}(p', R')$ which in turn means there is a substitution $\theta'$ such that 
    \[ \theta'(p') = \theta'(y_1) \dots \theta'(y_{m'})  = \theta(x_1) \ldots \theta(x_{m}) = \theta(p)\]
    Clearly, for each variable $x_i$ and each value of $\alpha\in\{1,\ldots,m'\}$, there are exactly $|[x_i]|$ occurrences of either the substring $ab^{(2m'i)+2\alpha-1}a ~ ab^{(2m'i)+2\alpha}a$ or the substring $ab^{(2m'i)+2\alpha}a ~ ab^{(2m'i)+2\alpha-1}a$. Since the substrings for these $m'$ choices of $\alpha$ occur $|[x_i]|$ times each in $\theta(p)$, and the decomposition of $\theta(p)$ into $\theta'(y_1)$, \dots, $\theta'(y_{m'})$ can split no more than $m'-1$ such substrings, there must be some $\alpha_i\in\{1,\ldots,m'\}$ such that \emph{every}\/ occurrence of the substrings $ab^{(2m'i)+2\alpha_i-1}a ~ ab^{(2m'i)+2\alpha_i}a$ and $ab^{(2m'i)+2\alpha_i}a ~ ab^{(2m'i)+2\alpha_i-1}a$ occurs as a substring of $\theta'(y_j)$ for some $j$ (i.e., no occurrence of these substrings is split by the decomposition of $\theta(p)$ into $\theta'(y_1)$, \dots, $\theta'(y_{m'})$). 

The desired reversal-obedient morphism $\phi$ is now easy to define. For each variable $y\in \mathit{Var}(p')$, we define $\phi(y)$ as follows. First, remove from $\theta'(y)$ all terminal symbols except for all full substrings of the form $ab^{(2m'i)+2\alpha_i-1}a ~ ab^{(2m'i)+2\alpha_i}a$ or $ab^{(2m'i)+2\alpha_i}a ~ ab^{(2m'i)+2\alpha_i-1}a$ for $\alpha_i$ as defined above. Second, replace each remaining full substring of the form $ab^{(2m'i)+2\alpha_i-1}a ~ ab^{(2m'i)+2\alpha_i}a$ by $x_i$, and each full substring of the form $ab^{(2m'i)+2\alpha_i}a ~ ab^{(2m'i)+2\alpha_i-1}a$ by $x_i^{rev}$. (If nothing remains, then $\theta'(y)=\varepsilon$.)
It is easy to see that $\phi$ is a reversal-obedient morphism for $R'$, and $\phi\overline{(p', R')} = \overline{(p,R)}$. 
\end{proof}

\begin{theorem}\label{thm-nonlearnable-rev-gold}
Let $|\Sigma| = 2$. Then $\varepsilon\mathcal{L}_{rev}$ does not possess a family of positive characteristic sets. 
\end{theorem}

\begin{proof}

Fix $\Sigma = \{a, b\}$. We will show that $\varepsilon L_{rev}(x_1 x^{rev}_1 x_2 x^{rev}_2 x_3 x^{rev}_3)$ does not have a positive characteristic set with respect to $\varepsilon\mathcal{L}_{rev}$. Our proof follows a construction used by Reidenbach~\cite{Reidenbach02} to show that $\varepsilon L_{eq}(x_1 x_1 x_2 x_2 x_3 x_3)$ does not have a telltale with respect to $\varepsilon\mathcal{L}_{eq}$.
Let $(p,R)= x_1 x^{rev}_1 x_2 x^{rev}_2 x_3 x^{rev}_3$. In particular, we will show that for any finite set $T \subseteq \varepsilon L_{rev}(p,R)$ there exists a relational pattern $(p',R')$ with $p'\in X^+$ such that $T \subseteq \varepsilon L_{rev}(p',R') \subset \varepsilon L_{rev}(p,R)$. This implies that $T\times\{1\}$ cannot be used as a positive characteristic set for $\varepsilon L_{rev}(p,R)$ with respect to $\varepsilon\mathcal{L}_{rev}$, which proves the theorem.

Fix $T=\{w_1, \ldots, w_n\} \subseteq \varepsilon L_{rev}(p,R)$. For each $w_i \in T$, define $\overleftarrow{\theta_i}: \Sigma^\ast \longrightarrow X^\ast$ by $\overleftarrow{\theta_i}(\sigma_1\ldots\sigma_z)=\overleftarrow{\theta_i}(\sigma_1)\ldots\overleftarrow{\theta_i}(\sigma_z)$, for $\sigma_1,\ldots,\sigma_z\in\Sigma$, where, for $c\in\Sigma$, we set
$$
\overleftarrow{\theta_i}(c) =\begin{cases}
			x_{2i-1}, & \text{$c=a$}\\
            x_{2i}, & \text{$c=b$}
		 \end{cases}
$$

$T \subseteq \varepsilon L_{rev}(p,R)$ implies for every $i \in \{1, \ldots, n\}$ the existence of a substitution $\theta_i$ (valid for $R$) such that $w_i = \theta_i(p)$. For each $w_i \in T$ we construct three strings of variables $\alpha_{i,1},\alpha_{i,2},\alpha_{i,3} \in X^\ast$, such that a concatenation of these strings in a specific way produces a pattern other than $(p,R)$ that generates $w_i$. These $\alpha_{i,k}$ will be the building blocks for the desired pattern $(p',R')$. Consider two cases:
\begin{itemize}
    \item[($i$)] Some $\sigma \in \Sigma$ appears in $\theta_i(x_3)$ exactly once while $\theta_i(x_1), \theta_i(x_2) \in \{\sigma'\}^\ast$ for $\sigma' \neq \sigma$, $\sigma'\in\Sigma$.
    Here we construct strings of variables as follows:
    
    $\alpha_{i,1} := \overleftarrow{\theta_i}(\theta_i(x_1) \theta_i(x_2))$
    
    $\alpha_{i,2} := \overleftarrow{\theta_i}(\theta_i(x_3))$
        
    $\alpha_{i,3} := \varepsilon$
    
    
    \item[($ii$)] Not ($i$). In this case we simply set $\alpha_{i,k} := \overleftarrow{\theta_i}(\theta_i(x_k))$ where $1 \leq k \leq 3$.
\end{itemize}
In each case, $w_i \in \varepsilon L_{rev}(\alpha_{i,1} \alpha^{rev}_{i,1} \alpha_{i,2} \alpha^{rev}_{i,2} \alpha_{i,3} \alpha^{rev}_{i,3})$.

Now define 
$(p',R') := y_1 y^{rev}_1 y_2 y^{rev}_2 y_3 y^{rev}_3$
where $y_k := \alpha_{1,k} \alpha_{2,k} \ldots \alpha_{n,k}\in X^+$ for $1 \leq k \leq 3$. 
To conclude the proof, we show $T \subseteq \varepsilon L_{rev}(p',R') \subset \varepsilon L_{rev}(p,R)$. 

To see that $T \subseteq \varepsilon L_{rev}(p',R')$, note $w_i \in \varepsilon L_{rev}(\alpha_{i,1} \alpha^{rev}_{i,1} \alpha_{i,2} \alpha^{rev}_{i,2} \alpha_{i,3} \alpha^{rev}_{i,3})$, and thus $w_i$ can be generated from $(p',R')$ by replacing all variables in $\alpha_{j,k}$, $j\ne i$, $1\le k\le 3$, with the empty string.

    
To verify $\varepsilon L_{rev}(p', R') \subseteq \varepsilon L_{rev}(p, R)$, it suffices to provide a reversal-obedient morphism $\Phi: X^\ast \longrightarrow X^\ast$ such that $\Phi\overline{(p,R)} = \overline{(p', R')}$. The existence of such $\Phi$ is obvious from our construction, namely $\Phi(x_k)= \alpha_{1,k} \alpha_{2,k} \ldots \alpha_{n,k}$ for $1 \leq k \leq 3$. 
    
Finally, we prove $\varepsilon L_{rev}(p', R')\subset\varepsilon L_{rev}(p,R)$. By way of contradiction, using Lemma~\ref{lem-cf-inclusion}, suppose there is a morphism $\Psi: X^\ast \longrightarrow X^\ast$ with $\Psi\overline{(p',R')} = \overline{(p,R)}$. 

By construction, $(p',R')$ has no free variables, i.e., the group of each variable in $p'$ has size at least $2$. Hence we can decompose $\Psi$ into two morphisms $\psi$ and $\psi'$ such that $\Psi\overline{(p', R')} = \psi'(\psi\overline{(p',R')}) = \overline{(p,R)}$ and, for each $v_j\in \mathit{Var}(p')$:\\ (i) $\psi(v_j) =\varepsilon$, if $|[v_j]|>2$, (ii) $\psi(v_j) =v_j$, if $|[v_j]|=2$.
    Since $\overline{(p', R')}=\Phi\overline{(p, R)}$ and either $[\Phi(x_3)]= \emptyset$ or $|[\Phi(x_3)]| \geq 4$ (see cases ($i$) and ($ii$)), we conclude that $\psi(\Phi(x_3))=\varepsilon$. Therefore,
    \begin{equation*}
        \begin{split}
           \psi\overline{(p', R')} & = \overbrace{v_{j_1} v_{j_2} \ldots v_{j_t}}^{\psi(\Phi(x_1))} \overbrace{v^{rev}_{j_t} v^{rev}_{j_{t-1}} \ldots v^{rev}_{j_1}}^{\psi(\Phi(x^{rev}_1))} \\
           & ~~~~ \overbrace{v_{j_{t+1}} v_{j_{t+2}} \ldots v_{j_{t+\ell}}}^{\psi(\Phi(x_2))} \overbrace{v^{rev}_{j_{t+\ell}} v^{rev}_{j_{t+\ell-1}} \ldots v^{rev}_{j_{t+1}}}^{\psi(\Phi(x^{rev}_2))} \overbrace{\varepsilon}^{\psi(\Phi(x_3))} \overbrace{\varepsilon}^{\psi(\Phi(x^{rev}_3))}
       \end{split}
    \end{equation*}
    where $t,\ell \geq 1$ and $v_{j_\zeta} \neq v_{j_\eta}$ for $\zeta \neq \eta, ~ 1\leq \zeta ,\eta \leq t + \ell$.
    To obtain a pattern with the same Parikh vector as of $(p, R)$, $\psi'$ must replace all except six variables with $\varepsilon$. 
    By the structure of $\psi\overline{(p', R')}$, the only possible options up to renaming of variables are $\overline{(p_1, R_1)} := x_1 x_2 x^{rev}_2 x^{rev}_1 x_3 x^{rev}_3$, $\overline{(p_2, R_2)} := x_1 x^{rev}_1 x_2 x_3 x^{rev}_3 x^{rev}_2$ and $\overline{(p_3, R_3)} := x_1 x_2 x_3 x^{rev}_3 x^{rev}_2 x^{rev}_1$. None of these are equivalent to $(p, R)$. 
    This is because all words of the form $ba^{2\kappa_1} bb a^{2\kappa_2} bb a^{2\kappa_3}b$ for pairwise distinct $\kappa_1, \kappa_2, \kappa_3 \in \mathbb{N}$ witness $\varepsilon L_{rev}(p, R) \not\subseteq \varepsilon L_{rev}(p_1, R_1)$, $\varepsilon L_{rev}(p, R) \not\subseteq \varepsilon L_{rev}(p_2, R_2)$, and $\varepsilon L_{rev}(p, R)\not \subseteq \varepsilon L_{rev}(p_3, R_3)$.

    Hence there is no morphism $\psi'$ such that $(p, R) = \psi'(\psi((p', R')))$, and no $\Psi$ such that $\Psi(\overline{(p', R')}) = \overline{(p, R)}$. This is a contradiction and therefore $\varepsilon L_{rev}(p', R')\subset\varepsilon L_{rev}(p, R)$.
\end{proof}


\begin{example}
    Consider the set $T:=\{w_1,w_2,w_3,w_4\}$ with 
    $$w_1 = \underbrace{a}_{\theta_1(x_1)} 
    \underbrace{a}_{\theta_1(x^{rev}_1)} 
    \underbrace{a}_{\theta_1(x_2)} 
    \underbrace{a}_{\theta_1(x^{rev}_2)} 
    \underbrace{b}_{\theta_1(x_3)} 
    \underbrace{b}_{\theta_1(x^{rev}_3)}$$
    $$w_2 = \underbrace{ba}_{\theta_2(x_1)} 
    \underbrace{ab}_{\theta_2(x^{rev}_1)} 
    \underbrace{a}_{\theta_2(x_2)} 
    \underbrace{a}_{\theta_2(x^{rev}_2)} 
    \underbrace{ab}_{\theta_2(x_3)} 
    \underbrace{ba}_{\theta_2(x^{rev}_3)}$$
    $$w_3 = \underbrace{b}_{\theta_3(x_1)} 
    \underbrace{b}_{\theta_3(x^{rev}_1)} 
    \underbrace{bb}_{\theta_3(x_2)} 
    \underbrace{bb}_{\theta_3(x^{rev}_2)} 
    \underbrace{bb}_{\theta_3(x_3)} 
    \underbrace{bb}_{\theta_3(x^{rev}_3)}$$
    $$w_4 = \underbrace{aab}_{\theta_4(x_1)} 
    \underbrace{baa}_{\theta_4(x^{rev}_1)} 
    \underbrace{b}_{\theta_4(x_2)} 
    \underbrace{b}_{\theta_4(x^{rev}_2)} 
    \underbrace{bab}_{\theta_4(x_3)} 
    \underbrace{bab}_{\theta_4(x^{rev}_3)}$$

    A pattern $\overline{(p', R')}$ generating these four words and generating a proper subset of $\varepsilon L_{rev}(x_1x_1^{rev}x_2x_2^{rev}x_3x_3^{rev})$ is constructed as follows. 
    $$\alpha_1\!=\!\underbrace{x_1x_1}_{\overleftarrow{\theta_1}(a ~ a)}^{\overbrace{}^{\alpha_{1,1}}} 
    \underbrace{x_1x_1}_{\overleftarrow{\theta_1}(a~ a)}^{\overbrace{}^{\bar{\alpha}_{1,1}}} 
    \underbrace{x_2}_{\overleftarrow{\theta_1}(b)}^{\overbrace{}^{\alpha_{1,2}}} 
    \underbrace{x_2}_{\overleftarrow{\theta_1}(b)}^{\overbrace{}^{\bar{\alpha}_{1,2}}} 
    {\overbrace{\varepsilon}^{\alpha_{1,3}}} 
    {\overbrace{\varepsilon}^{\bar{\alpha}_{1,3}}}\,\ \ \alpha_2\!=\!\underbrace{x_4 x_3}_{\overleftarrow{\theta_2}(ba)}^{\overbrace{}^{\alpha_{2,1}}}
    \underbrace{x_3 x_4}_{\overleftarrow{\theta_2}(ab)}^{\overbrace{}^{\bar{\alpha}_{2,1}}}
    \underbrace{x_3}_{\overleftarrow{\theta_2}(a)}^{\overbrace{}^{\alpha_{2,2}}}
    \underbrace{x_3}_{\overleftarrow{\theta_2}(a)}^{\overbrace{}^{\bar{\alpha}_{2,2}}}
    \underbrace{x_3 x_4}_{\overleftarrow{\theta_2}(ab)}^{\overbrace{}^{\alpha_{2,3}}} \underbrace{x_4 x_3}_{\overleftarrow{\theta_2}(ba)}^{\overbrace{}^{\bar{\alpha}_{2,3}}}$$
    $$\alpha_3\!=\!\underbrace{x_6}_{\overleftarrow{\theta_3}(b)}^{\overbrace{}^{\alpha_{3,1}}} 
    \underbrace{x_6}_{\overleftarrow{\theta_3}(b)}^{\overbrace{}^{\bar{\alpha}_{3,1}}} 
    \underbrace{x_6 x_6}_{\overleftarrow{\theta_3}(bb)}^{\overbrace{}^{\alpha_{3,2}}} 
    \underbrace{x_6 x_6}_{\overleftarrow{\theta_3}(bb)}^{\overbrace{}^{\bar{\alpha}_{3,2}}} 
    \underbrace{x_6 x_6}_{\overleftarrow{\theta_3}(bb)}^{\overbrace{}^{\alpha_{3,3}}} 
    \underbrace{x_6 x_6}_{\overleftarrow{\theta_3}(bb)}^{\overbrace{}^{\bar{\alpha}_{3,3}}}\,\ \ \alpha_4\!=\!\underbrace{x_7 x_7 x_8}_{\overleftarrow{\theta_4}(aab)}^{\overbrace{}^{\alpha_{4,1}}} 
    \underbrace{x_8 x_7 x_7}_{\overleftarrow{\theta_4}(baa)}^{\overbrace{}^{\bar{\alpha}_{4,1}}} 
    \underbrace{x_8}_{\overleftarrow{\theta_4}(b)}^{\overbrace{}^{\alpha_{4,2}}} 
    \underbrace{x_8}_{\overleftarrow{\theta_4}(b)}^{\overbrace{}^{\bar{\alpha}_{4,2}}}
    \underbrace{x_8 x_7 x_8}_{\overleftarrow{\theta_4}(bab)}^{\overbrace{}^{\alpha_{4,3}}}
    \underbrace{x_8 x_7 x_8}_{\overleftarrow{\theta_4}(bab)}^{\overbrace{}^{\bar{\alpha}_{4,3}}}$$

    Hence, the pattern $\overline{(p', R')}$ is of the form:

    \begin{align*}
        ~~~~~\overline{(p', R')}&= \underbrace{\overbrace{x_1x_1}^{\alpha_{1,1}}        \overbrace{x_4 x_3}^{\alpha_{2,1}} \overbrace{x_6}^{\alpha_{3,1}} \overbrace{x_7 x_7 x_8}^{\alpha_{4,1}}}_{\Phi(x_1)} 
        \underbrace{\overbrace{x^{rev}_8 x^{rev}_7 x^{rev}_7}^{\alpha^{rev}_{4,1}} \overbrace{x^{rev}_6}^{\alpha^{rev}_{3,1}} \overbrace{x^{rev}_3 x^{rev}_4}^{\alpha^{rev}_{2,1}} \overbrace{x^{rev}_1x^{rev}_1}^{\alpha^{rev}_{1,1}}}_{\Phi(x^{rev}_1)}\\ 
        &~~~\underbrace{\overbrace{x_2}^{\alpha_{1,2}}\overbrace{x_3}^{\alpha_{2,2}} \overbrace{x_6 x_6}^{\alpha_{3,2}} \overbrace{x_8}^{\alpha_{4,2}}}_{\Phi(x_2)} 
        \underbrace{\overbrace{x^{rev}_8}^{\alpha^{rev}_{4,2}} \overbrace{x^{rev}_6 x^{rev}_6}^{\alpha^{rev}_{3,2}} \overbrace{x^{rev}_3}^{\alpha^{rev}_{2,2}} \overbrace{x^{rev}_2}^{\alpha^{rev}_{1,2}}}_{\Phi(x^{rev}_2)}\\
        &~~~
        \underbrace{\overbrace{x_3 x_4}^{\alpha_{2,3}} \overbrace{x_6 x_6}^{\alpha_{3,3}} \overbrace{x_8 x_7 x_8}^{\alpha_{4,3}}}_{\Phi(x_3)} 
        \underbrace{\overbrace{x^{rev}_8 x^{rev}_7 x^{rev}_8}^{\alpha^{rev}_{4,3}} \overbrace{x^{rev}_6 x^{rev}_6}^{\alpha^{rev}_{3,3}} \overbrace{x^{rev}_4 x^{rev}_3}^{\alpha^{rev}_{2,3}}}_{\Phi(x^{rev}_3)} 
    \end{align*}

    Obviously, $T= \{w_1, w_2, w_3, w_4\} \subseteq \varepsilon L_{rev}(p', R')$ and $\varepsilon L_{rev}(p', R') \subseteq \varepsilon L_{rev}(p, R)$. However, $aa~bb~aaaa \in \varepsilon L_{rev}(p, R) \setminus \varepsilon L_{rev}(p', R')$.

\end{example}

\section{Equal-Length Relation}

This section assumes $r=len$, i.e., variables in relation are replaced by words of equal length, independent of the actual symbols in those words. Recent work by Holte et al.~\cite{HolteMZ22} studied this relation for non-erasing pattern languages. We will demonstrate how one of their main results can be used to show that substitutions replacing each variable with a word of length at most 2 form positive characteristic sets, if the underlying alphabet has at least three symbols. For this purpose, we introduce the following notation:
\begin{enumerate}
    \item $L_{len}(p,R)$ refers to the non-erasing language generated by $(p,R)$, which is the subset of $\varepsilon L_{len}(p,R)$ that results from applying only valid $R$-substitutions $\varphi$ to $p$ that do not erase any variable, i.e., $|\varphi(x)|\ge 1$ for all $x\in\mathit{Var}(p)$.

    \item $\mathcal{S}_2(p,R)$ denotes the set of words in $L_{len}(p,R)$ that are generated by $R$-substitutions $\varphi$ satisfying $\exists x\in \mathit{Var}(p)\ [\forall y\in [x]\ |\varphi(y)|= 2\ \wedge\ \forall y\in \mathit{Var}(p)\setminus [x]\ |\varphi(y)|=1]$.
\end{enumerate}

Holte et al.\ showed that, if $|\Sigma|\ge 3$ and $\mathcal{S}_2(p,R)\subseteq L_{len}(p,R')$, then $L_{len}(p,R)\subseteq L_{len}(p,R')$. Note that, while the size of $\mathcal{S}_2(p,R)$ is in general not linear in $|p|$, the actual subsets  $\mathcal{S}_2(p,R)$ used in the proof by Holte et al.\ are of size linear in $|p|$. Moreover, all words in $\mathcal{S}_2(p,R)$ are generated from $p$ using an $\ell_2$-substitution, in the sense of the following definition.

\smallskip

\begin{definition}
Let $z \in \mathbb{N}$. Any substitution that maps each variable to a word of length at most $z$ is called an $\ell_z$-substitution. Given a relational pattern $(p,R)$, the set $L_{len}(p,R,z)\subseteq L_{len}(p,R)$ is the set of all words generated from $p$ using non-erasing $R$-substitutions that are also $\ell_z$-substitutions. When dropping the condition ``non-erasing'', we obtain a subset of $\varepsilon L_{len}(p,R)$, which we denote with $\varepsilon L_{len}(p,R,z)$.
\end{definition}

\smallskip

One can easily conclude from Holte et al.'s result that the following theorem holds for general pairs of relational patterns $((p,R),(p',R'))$, where $p\ne p'$ is allowed.

\smallskip

\begin{theorem}\label{thm:Sigmage3NE}
    Let $|\Sigma|\ge 3$. If $\mathcal{S}_{2}(p,R)\subseteq L_{len}(p',R')$ and $\mathcal{S}_{2}(p',R')\subseteq L_{len}(p,R)$, then $L_{len}(p,R)= L_{len}(p',R')$. In particular, $\mathcal{S}_2(p,R)$ is a positive characteristic set for $L_{len}(p,R)$ with respect to~$\mathcal{L}_{len}$.
\end{theorem}

\smallskip

The subsets of $\mathcal{S}_2(p,R)$ used in the proof are actually of size linear in $|p|$, so that the languages in $\mathcal{L}_{len}$ have linear-size  positive characteristic sets. Moreover, these sets can be effectively constructed as follows.

For each variable  $x$ occurring in $p$, let $\sigma_x\in\Sigma$ be a letter distinct from the last symbol in the nearest terminal block to the left of $x$ in $p$ (if such terminal block exists), and distinct from the first symbol in the nearest terminal block to the right of $x$ in $p$ (if such terminal block exists). Since $|\Sigma|\ge 3$, such a letter $\sigma_x$ always exists. Further, define a substitution $\theta_x$ by 
$$
\theta_x(y):=\begin{cases}
    \sigma_y\sigma_y\,&\mbox{if }y\in[x]\,,\\
    \sigma_y\,,&\mbox{otherwise}\,,
\end{cases}
$$
for all $y\in\mathit{Var}(p)$, and a substitution $\theta$ by $\theta(y)=\sigma_y$ for all $y\in\mathit{Var}(p)$.
Then, as follows from the proof details presented by Holte et al.~\cite{HolteMZ22}, \( \{ \theta_x(p) \mid x\in\mathit{Var}(p) \} \cup\{\theta(p)\}\subset \mathcal{S}_2(p,R) \) is a linear-size positive characteristic set for $L_{len}(p,R)$ with respect to~$\mathcal{L}_{len}$.

As a consequence, under the premises of Theorem~\ref{thm:Sigmage3NE}, the equivalence of two relational patterns can be decided by testing membership of a small number of short strings (relative to the length of the underlying patterns).

We will first show that inclusion of two relational pattern languages cannot be decided in the same fashion. In particular, for deciding inclusion, it is not even sufficient to test \emph{all}\/ valid non-erasing $\ell_2$-substitutions (rather than just those generating $\mathcal{S}_2(p,R)$).

\smallskip

\begin{theorem}\label{thm:Sigmage3NEinclusion}
    Let $|\Sigma|\ge 2$. There are relational patterns $(p,R)$ and $(p',R')$ such that $L_{len}(p,R,2)\subseteq L_{len}(p',R')$, but $L_{len}(p,R)\not\subseteq L_{len}(p',R')$. 
\end{theorem}

\begin{proof}
Let \( p = x_1 x_2 \ ab\ y_1 y_2 y_3\ ab\ z_1 z_2 z_3 z_4 z_5 \) and \( p' = x'_1 x'_2 ab z'_1 z'_2 z'_3 z'_4 z'_5 \) where all shown symbols other than $a,b$ are variables and \[ R=\{ (x_1, x_2), (y_1, y_2), (y_2, y_3), (z_1, z_2), (z_2, z_3), (z_3, z_4), (z_4, z_5)\} ,\ \] \[ R' =\{ (x'_1, x'_2), (z'_1, z'_2), (z'_2, z'_3), (z'_3, z'_4), (z'_4, z'_5)\} .\ \] 
Thus, in each block of each pattern, all the variables within the block form a group.

To see that $L_{len}(p,R,2)\subseteq L_{len}(p',R')$, let $\theta$ be any non-erasing $R$-substitution that replaces variables by words of length at most 2. Then $\theta(p)$ is of the form $w_1abw_2abw_3$ where $|w_1|=2k_1$, $|w_2|=3k_2$, and $|w_3|=5k_3$, for $k_i\in\{1,2\}$. If $k_2=1$, then $3k_2+5k_3+2=5(k_3+1)$, and thus $\theta(p)$ can be generated from $p'$ by a non-erasing $R'$-substitution that replaces the $x'_i$ by suitable words of length $k_1$, and the $z'_j$ by suitable words of length $k_3+1$. If $k_2=2$, then $2k_1+2+3k_2=2k_1+8=2(k_1+4)$. Therefore, $\theta(p)$ can be generated from $p'$ by a non-erasing $R'$-substitution that replaces the $x'_i$ by suitable words of length $k_1+4$, and each $z'_j$ by~$\theta(z_j)$.

In order to verify $L_{len}(p,R)\not\subseteq L_{len}(p',R')$, consider the word $a^2\ ab\ a^9 \ ab\ a^5$, which clearly belongs to $L_{len}(p,R)$.
Fix any non-erasing $R'$-substitution $\theta'$. Since $|\theta'(x'_1 x'_2)|$ is even, we know that $\theta'(x'_1 x'_2) \ne a^2\ ab\ a^9$. Likewise, since $|\theta'(z'_1 z'_2 z'_3 z'_4 z'_5)|$ is a multiple of $5$, we have  $\theta'(z'_1 z'_2 z'_3 z'_4 z'_5) \ne a^9\ ab\ a^5$. Thus $\theta'(p')\ne\theta(p)$, which concludes the proof.
\end{proof}

For the erasing case, we obtain an analogous negative result on deciding inclusion. The construction in the proof is similar, but requires a slightly different design.

\smallskip

\begin{theorem}
    Let $|\Sigma|\ge 2$. There are relational patterns $(p,R)$ and $(p',R')$ such that $\varepsilon L_{len}(p,R,2)\subseteq \varepsilon L_{len}(p',R')$, but $\varepsilon L_{len}(p,R)\not\subseteq \varepsilon L_{len}(p',R')$. 
\end{theorem}

\begin{proof}
 Let \[ p = ab^{100} \ y_1 y_2\ z_1 z_2 z_3\ ab^{100} \, , \]  \[ p' = \vec{x}^{(103)}_{1}\ \vec{x}^{(104)}_{2}\ \vec{x}^{(105)}_{3}\ \vec{x}^{(106)}_{4}\ \vec{x}^{(107)}_{5}\  \vec{x}^{(108)}_{6}\ \vec{x}^{(109)}_{7}\ \vec{x}^{(111)}_{8}\ ab^{100}\ \vec{x}^{(101)}_{9}\,, \] where \( \vec{x}^{(j)}_i = x_{i,1}
    \dots x_{i,j} \) is a sequence of $j$ variables. Moreover, let \[R=\{ (y_1,y_2), (z_1,z_2),(z_2,z_3)\} ,\ \] 
    \[ R' =\{ (x_{i,j},x_{i,j'})\mid 1\le i\le 9 , 1\le j,j' \le\jmath_i\}\,,\ \]
    where $\jmath_1=103$, $\jmath_2=104$, $\jmath_3=105$, $\jmath_4=106$,  $\jmath_5=107$, $\jmath_6=108$, $\jmath_7=109$, $\jmath_8=111$, $\jmath_9=101$.

    To see that $\varepsilon L_{len}(p,R,2)\subseteq \varepsilon L_{len}(p',R')$, let $\theta$ be any $R$-substitution that replaces variables by words of length at most 2. Then $\theta(p)$ is of the form $ab^{100}w_1w_2ab^{100}$ where $|w_1|=2k_1$ and $|w_2|=3k_2$, for $k_i\in\{0,1,2\}$. Thus, the length $|w_1w_2|$ could take the values 0, 2, 3, 4, 5, 6, 7, 8, or 10. Consequently, the length $|ab^{100}w_1w_2|=|w_1w_2ab^{100}|$ could take the values 101, 103, 104, 105, 106, 107, 108, 109, or 111.  If $|ab^{100}w_1w_2|\in\{103, 104, 105, 106, 107, 108, 109, 111\}$, then  consider an $R'$-substitution $\theta'$ that, for some suitably chosen $i\in\{1,\ldots,8\}$, replaces the variables of the form $x_{i,j}$ with suitable strings of length 1, so that $\theta'(\vec{x}_i^{(\jmath_i)})=ab^{100}w_1w_2$; moreover, $\theta'$ erases all other variables. Then $\theta'(p')=\theta(p)$. If $|w_1w_2ab^{100}| (=|ab^{100}w_1w_2|)=101$, then consider an $R'$-substitution $\theta'$ that erases all variables of the form $x_{i,j}$ for $1\le i\le 8$ and replaces the variables of the form $x_{9,j}$ with suitable strings of length 2, so that $\theta'(\vec{x}_9^{(101)})=w_1w_2ab^{100}$. Here as well we obtain $\theta'(p')=\theta(p)$. Therefore, $\varepsilon L_{len}(p,R,2)\subseteq \varepsilon L_{len}(p',R')$.

    In order to verify $\varepsilon L_{len}(p,R)\not\subseteq \varepsilon L_{len}(p',R')$, consider the word $ab^{100}\ a^9\ ab^{100}$, which is clearly a member of $\varepsilon L_{len}(p,R)$. 
    Since $|ab^{100} a^9|=|a^9ab^{100}|=110$, we obtain $ab^{100}\ a^9\ ab^{100}\notin\varepsilon L_{len}(p',R')$, which concludes the proof.
\end{proof}


\subsection{Alphabets of Size at Least Three}

For learning-theoretic purposes, by contrast with $r=eq$ or $r=rev$, there is no real technical difference between the erasing and the non-erasing case when $r=len$. In particular, Theorem~\ref{thm:Sigmage3NE}, which was proven by Holte et al.~\cite{HolteMZ22}, carries over to the erasing case, using only a few technical changes to the proof. Where substitutions replacing variables with strings of length either 1 or 2 were crucial in Holte et al.'s result, in the erasing case, it is sufficient to consider substitutions replacing variables with strings of length either 0 or 1.
To formulate this result, we first define, for any relational pattern $(p,R)$ and any $\ell \in \mathbb{N}\setminus \{0\}$:
\[\mathcal{S}_{\varepsilon,\ell}(p,R)=\{\varphi(p)\mid\exists x\in \mathit{Var}(p)\ \exists z\le \ell\ [\forall y\in[x]\ |\varphi(y)|= z \wedge \forall y\notin[x]\ \varphi(y)=\varepsilon]\}\,.\]

The analogous version of Theorem~\ref{thm:Sigmage3NE} for the erasing case now reads as follows.

\smallskip

\begin{theorem}\label{prop:Sigmage3}
    Let $|\Sigma|\ge 3$. If $\mathcal{S}_{\varepsilon,1}(p,R)\subseteq\varepsilon L_{len}(p',R')$ and $\mathcal{S}_{\varepsilon,1}(p',R')\subseteq\varepsilon L_{len}(p,R)$, then $\varepsilon L_{len}(p,R) = \varepsilon L_{len}(p',R')$. In particular, $\mathcal{S}_{\varepsilon,1}(p,R)$ is a positive characteristic set for $\varepsilon L_{len}(p,R)$ with respect to~$\varepsilon\mathcal{L}_{len}$.
\end{theorem}

\smallskip

A close inspection of the proof shows again that the actual subsets of $\mathcal{S}_{\varepsilon,1}(p,R)$ used in the construction are of size linear in~$|p|$.

To prove this theorem, we adjust Holte et al.'s proof from the non-erasing case, which requires a few technical definitions and lemmas. We begin with the erasing version of ~\cite[Lemma 8]{HolteMZ22}; its proof is straightforward:

\smallskip

 \begin{lemma}\label{swap-variables_length}
    Suppose $(p, R)$ is a relational pattern, where $p = q_1 x_i x_{i+1} q_2$ with $q_1, q_2 \in (\Sigma \cup X)^\ast$. Then $\varepsilon L_{len} (p, R) = \varepsilon L_{len} (p, R')$, where $R'$ is obtained from R by swapping $x_i$ with $x_{i+1}$ in all relations.
 \end{lemma}
 
\smallskip

\begin{definition}
    Let $(p, R)$ be a relational pattern, and let $(\vec{x}_1$, \dots, $\vec{x}_n)$ be the list of non-empty variable blocks in~$p$, ordered from left to right according to their appearance in $p$. Let $[x]$ be a group in $(p, R)$. The decomposition of $[x]$ over $(p,R)$ is the vector $\langle |\mathit{Var}(\vec{x}_1)\cap[x]|, \dots, |\mathit{Var}(\vec{x}_n)\cap[x]| \rangle$.
\end{definition}

\smallskip

The following lemma is now straightforward.

\smallskip

\begin{lemma}\label{lem:linearcombination}
Let $(\bar{p},R)$ be any relational pattern, and let $p$ be the pattern obtained by iteratively removing from $\bar{p}$, in a one at a time manner, all variable groups whose decomposition is a positive integer linear combination of the decomposition of other groups in $(\bar{p},R)$. Then $\varepsilon L_{len}(\bar{p}, R) = \varepsilon L_{len}(p, R)$.
\end{lemma}

\smallskip

Lemmas~\ref{swap-variables_length} and \ref{lem:linearcombination} allow us to define a normal form representation of relational patterns over the relation~$len$.

\smallskip

 \begin{definition}\label{equal-length normal form}
 Suppose $(p,R)$ is a relational pattern. Let $p=\vec{x}_1 \omega_1 \vec{x}_2 \omega_2 \ldots \vec{x}_n \omega_n \vec{x}_{n+1}$, with $n\ge 0$, $\vec{x}_1,\vec{x}_{n+1}\in X^*$, $\vec{x}_2,\ldots,\vec{x}_{n}\in X^+$, and $\omega_1,\ldots,\omega_n\in\Sigma^+$.\footnote{Note that $n=0$ results in a  pattern in $X^+$.}
    Let $([y_1], \ldots, [y_\kappa])$ be an ordered list of all groups in $(p,R)$.\footnote{We assume a canonical way of choosing such order.} The \emph{equal-length normal form} of $(p,R)$, denoted by $(p_{\nf},R)$, is obtained in two steps. First, one forms the pattern
     $$\bar{p}= \vec{x}_1(1)\ldots\vec{x}_1(\kappa) \omega_1 \ldots \omega_{n-1} \vec{x}_n(1)\ldots\vec{x}_n(\kappa) \omega_n \vec{x}_{n+1}(1)\ldots\vec{x}_{n+1}(\kappa)\,,$$
     where $\vec{x}_i(j)$ is the string of all variables in  $\mathit{Var}(\vec{x}_i) \cap [y_j]$, in increasing order of their indices in $X=\{x_1,x_2,\ldots\}$.\footnote{Since $([y_1], \ldots, [y_\kappa])$ is a canonically ordered list, this string is unique.} Second, one obtains $p_{\nf}$ by iteratively removing from $\bar{p}$, in a one at a time manner, all variable groups whose decomposition is a positive integer linear combination of the decomposition of other groups in $(\bar{p},R)$.
 \end{definition}
\smallskip

\begin{corollary}\label{cor:nf}
Let $(p,R)$ be any relational pattern. Then
    $\varepsilon L_{len}(p,R) = \varepsilon L_{len}(p_{\nf},R)$.
\end{corollary}

\smallskip

\begin{proof}
    Immediate from Lemmas \ref{swap-variables_length} and~\ref{lem:linearcombination}.
\end{proof}

\smallskip

\begin{example}
    Let $p=x_1x_2y_1y_2y_3y_4\ a\ x_3\ b\ x_4x_5x_6y_5y_6\ bba\ x_7 x_8 x_9y_7y_8\ aa$, $R=\{(x_1,x_2),(x_1,x_8),$ $(x_3,x_7),(x_5,x_8),(x_4,x_6)\}\cup\{(y_1,y_i)\mid 2\le i\le 8\}$. The groups in $(p,R)$ are $\{x_1,x_2,x_5,x_8\}$, $\{x_3,x_7\}$, $\{x_4,x_6\}$, $\{x_9\}$, $\{y_1,\ldots,y_8\}$. The decomposition of $[y_1]$ is $\langle 4,0,2,2\rangle=2\langle 2,0,1,1\rangle$, where $\langle 2,0,1,1\rangle$ is the decomposition of $[x_1]$. Thus $p_{\nf}=x_1x_2\ a\ x_3\ b\ x_5x_4x_6\ bba\ x_8 x_7 x_9\ aa$. In particular, $\varepsilon L_{len}(p,R) = \varepsilon L_{len}(p_{\nf},R)$.
 \end{example}

\smallskip

The normal form of a relational pattern will help verify Theorem~\ref{prop:Sigmage3}; to this end we first prove two more lemmas.

\begin{lemma}\label{prop:Sigma3_constant_block_matching}
Suppose $|\Sigma| \geq 3$, $\mathcal{S}_{\varepsilon,1}(p,R)\subseteq\varepsilon L_{len}(p',R')$, and $\mathcal{S}_{\varepsilon,1}(p',R')\subseteq\varepsilon L_{len}(p,R)$. Let $\omega_1\ldots,\omega_n,\omega'_1,\ldots,\omega'_{n'}\in\Sigma^+$, $\vec{x}_i,\vec{x}'_j\in X^+$ for $2\le i\le n$ and $2\le j\le n'$, as well as $\vec{x}_1,\vec{x}_{n+1},\vec{x}'_1,\vec{x}'_{n'+1}\in X^*$ such that $p = \vec{x}_1 \omega_1 \ldots \vec{x}_n \omega_n \vec{x}_{n+1}$ and $p' = \vec{x}'_1 \omega'_1 \ldots \vec{x}'_{n'} \omega'_{n'} \vec{x}'_{n'+1}$. Then $n=n'$ and $\omega'_i = \omega_i$ for all $i \in \{1, \ldots, n\}$.
\end{lemma}

\begin{proof} The premises of the lemma imply $\omega_1 \cdots \omega_n\in \varepsilon L(p',R')$ and  $\omega'_1 \cdots \omega'_{n'}\in \varepsilon L(p,R)$. This is only possible if $\omega_1 \cdots \omega_n = \omega'_1 \cdots \omega'_{n'}$. 

Now suppose $(\omega_1, \ldots, \omega_n) \neq (\omega'_1, \ldots, \omega'_{n'})$. Since $\omega_1 \cdots \omega_n = \omega'_1 \cdots \omega'_{n'}$, there exists a smallest $i < \min\{n,n'\}$ such that $\omega'_i \neq \omega_i$. In particular, either $\omega'_i$ is a proper prefix of $\omega_i$ or vice versa. Without loss of generality, suppose $\omega'_i$ is a proper prefix of $\omega_i$. Thus, there is a string $\omega\in\Sigma^+$ such that $\omega'_i\omega=\omega_i$. Now let $\sigma\in\Sigma$ be a letter different from the last letter of $\omega'_i$ and different from the first letter of $\omega$. Such $\sigma$ exists, since $|\Sigma|\ge 3$. Further, let $x$ be a variable occurring in~$\vec{x}'_{i+1}$.

Now define a substitution $\theta$ by $\theta(y)=\sigma$ if $y\in [x]_{R'}$, and $\theta(y)=\varepsilon$ otherwise. Clearly, $\theta(p')\in\mathcal{S}_{\varepsilon,1}(p',R')$. Since $\mathcal{S}_{\varepsilon,1}(p',R')\subseteq\varepsilon L_{len}(p,R)$, we obtain $\theta(p')\in\varepsilon L_{len}(p,R)$. Thus, there is an $R$-substitution $\theta'$ such that $\theta'(p)=\theta(p')$. Let $\sigma'$ be the last symbol of $\omega'_i$ and let $k$ be such that the $k$th occurrence of the symbol $\sigma'$ in $\omega'_1\cdots\omega'_i$ is exactly in the last position of $\omega'_1\cdots\omega'_i$. Since $p$ and $p'$ contain the same terminal symbols in the same multiplicities, $\theta'(z)\in\{\sigma\}^*$ for each variable $z\in\textit{Var}(p)$. Therefore, the $k$th occurrence of $\sigma'$ in $\theta(p')$ is followed by $\sigma$ in $\theta(p')$. The same must be true for $\theta'(p)$, since $\theta'(p)=\theta(p')$. But in $\theta'(p)$, the $k$th occurrence of $\sigma'$ in $\theta(p')$ is followed by the first symbol of $\omega$, which differs from $\sigma$, by the choice of $\sigma$. This is a contradiction. 
\end{proof}

\begin{lemma}\label{prop:Sigma3_varBlocls_identical}
    Let $|\Sigma| \geq 3$, and let $(p,R),(p',R')$ be any two relational patterns. Suppose $\mathcal{S}_{\varepsilon,1}(p,R)\subseteq\varepsilon L_{len}(p',R')$ and $\mathcal{S}_{\varepsilon,1}(p',R')\subseteq\varepsilon L_{len}(p,R)$. Then, $p_{\nf} = p'_{\nf}$, subject to renaming of variables. 
\end{lemma}

\begin{proof}
    Lemma \ref{prop:Sigma3_constant_block_matching} implies that there exist $n\in\mathbb{N}$, $\omega_1$, \dots, $\omega_n\in\Sigma^*$, $\vec{x}_i$, $\vec{x}'_i\in X^*$ for $i\in\{1,n+1\}$, and $\vec{x}_j$, $\vec{x}'_j\in X^+$ for $2\le j\le n$, such that $p_{\nf} = \vec{x}_1 \omega_1 \ldots \vec{x}_n \omega_n \vec{x}_{n+1}$ and $p'_{\nf} = \vec{x}'_1 \omega_1 \ldots \vec{x}'_{n} \omega_{n} \vec{x}'_{n+1}$. Using the same notation as in Definition~\ref{equal-length normal form}, suppose $p_{\nf}$ contains $\kappa$ distinct groups $[y_1],\ldots,[y_\kappa]$, and $p'_{\nf}$ contains $\kappa'$ distinct groups $[y'_1],\ldots, [y'_{\kappa'}]$. 
   
   For $ 1 \leq i \leq \kappa$ define an $R$-substitution $\theta_i$ as follows. For each variable $y\in [y_i]$, if $y$ appears in $\vec{x}_j(i)$, then \( \theta_i(y) = c \) for some letter $c\in\Sigma$ that differs from both the last letter of $\omega_{j-1}$ and the first letter of $\omega_{j}$. (In case $j=1$, we only need $c$ to differ from the first letter of $\omega_1$, and in case $j=n+1$, $c$ differs from the last letter of $\omega_n$.)
   $\theta_i$ erases all other variables. Since \( \theta_1(p), \dots, \theta_\kappa(p) \in \mathcal{S}_{\varepsilon,1}(p,R) \) and \( \mathcal{S}_{\varepsilon,1}(p,R)\subseteq \varepsilon L_{len}(p',R') \), we have $\theta_i(p) \in \varepsilon L_{len}(p',R')$ for each $i$. Now, fix $i$ and suppose $\theta'_i$ is an $R'$-substitution such that $\theta_i(p_{\nf}) = \theta'_i(p'_{\nf})$. We want to show that $\theta_i(\vec{x}_j(i)) = \theta'_i(\vec{x}'_j)$ for all $j$. Note that
    \[ \theta_i(p_{\nf})\ =\ \theta_i(\vec{x}_1(i))\ \omega_1\ \theta_i(\vec{x}_2(i))\ \omega_2\ \dots\ \theta_i(\vec{x}_n(i))\ \omega_n\ \theta_i(\vec{x}_{n+1}(i))\,. \]

    Since $\theta_i(\vec{x}_1(i))$ does not contain any letter equal to the first symbol of $\omega_1$, we obtain that $\theta_i(\vec{x}_1(i))$ is a prefix of $\theta'_i(\vec{x}'_1)$ which in turn implies that $\theta_i(\vec{x}_1)\omega_1$ is a prefix of $\theta'_i(\vec{x}'_1)\omega_1$. 
    Consequently, and since $\theta_i(\vec{x}_2(i))$ does not contain any letter equal to the first symbol of~$\omega_2$, $\theta_i(\vec{x}_1(i))\omega_1\theta_i(\vec{x}_2(i))$ is a prefix of $\theta'_i(\vec{x}'_1)\omega_1\theta'_i(\vec{x}'_2)$ so that  $\theta_i(\vec{x}_1) \omega_1\theta_i(\vec{x}_2) \omega_2 $ is a prefix of $\theta'_i(\vec{x}'_1)\omega_1\theta'_i(\vec{x}'_2) \omega_2$. 
    Iterating this argument yields that $\theta_i(\vec{x}_1(i))\omega_1\ldots\omega_{j-1}\theta_i(\vec{x}_j(i))$ is a prefix of $\theta'_i(\vec{x}'_1)\omega_1\ldots\omega_{j-1}\theta'_i(\vec{x}'_j)$ for all $j$ and therefore $\theta_i(\vec{x}_1)\omega_1\ldots\omega_{j-1}\theta_i(\vec{x}_j) \omega_j$ is a prefix of $\theta'_i(\vec{x}'_1)\omega_1\ldots\omega_{j-1}\theta'_i(\vec{x}'_j) \omega_j$ for all~$j$.

    Since $\theta_i(\vec{x}_{n+1}(i))$ does not contain any letter equal to the last symbol of~$\omega_n$, we obtain that $\theta_i(\vec{x}_{n+1}(i))$ is a suffix of $\theta'_i(\vec{x}'_{n+1})$ and hence $\omega_n \theta_i(\vec{x}_{n+1})$ is a suffix of $\omega_n \theta'_i(\vec{x}'_{n+1})$. Iterating this argument like above yields that $\theta_i(\vec{x}_j(i))\omega_j\ldots\omega_n\theta_i(\vec{x}_{n+1}(i))$ is a suffix of $\theta'_i(\vec{x}'_j)\omega_j\ldots\omega_n\theta'_i(\vec{x}'_{n+1})$ for all $j$, and therefore $\omega_{j-1} \theta_i(\vec{x}_j)\omega_j\ldots\omega_n\theta_i(\vec{x}_{n+1})$ is a suffix of $\omega_{j-1} \theta'_i(\vec{x}'_j)\omega_j\ldots\omega_n\theta'_i(\vec{x}'_{n+1})$ for all~$j$. 

    In combination, this yields $\theta_i(\vec{x}_j(i)) = \theta_i(\vec{x}_j) = \theta'_i(\vec{x}'_j)$ for all~$j$.

    Therefore, for each group $[y_i]$ in $(p,R)$ there exist groups $[y'_{i_1}], \dots, [y'_{i_{k_i}}]$ in $(p', R')$ such that 
    \[ \langle\ \vec{x}_1(i),\ \vec{x}_2(i),\ \dots,\ \vec{x}_{n+1}(i)\ \rangle\ = \sum_{\jmath = 1}^{k_i}\ \lambda_\jmath\ \langle\ \vec{x}'_1(i_\jmath),\ \vec{x}'_2(i_\jmath),\ \dots,\ \vec{x}'_{n+1}(i_\jmath)\ \rangle \]
    where $\lambda_\jmath \in \mathbb{N}^+$.

    By a completely symmetric argument, for each group $[y'_i]$ in $(p',R')$ there exist groups $[y_{i_1}], \dots, [y_{i_{k_i}}]$ in $(p,R)$ such that 
    \[ \langle\ \vec{x}'_1(i),\ \vec{x}'_2(i),\ \dots,\ \vec{x}'_{n+1}(i)\ \rangle\ = \sum_{\jmath = 1}^{k_i}\ \lambda_\jmath\ \langle\ \vec{x}_1(i_\jmath),\ \vec{x}_2(i_\jmath),\ \dots,\ \vec{x}_{n+1}(i_\jmath)\ \rangle \]
    where $\lambda_\jmath \in \mathbb{N}^+$.

    In conclusion, this yields $p_{\nf} = p'_{\nf}$.  
\end{proof}

Now Theorem~\ref{prop:Sigmage3} follows from Lemma~\ref{prop:Sigma3_varBlocls_identical} and Corollary~\ref{cor:nf}.

\subsection{Binary Alphabets}

 Both in our Theorem~\ref{prop:Sigmage3} and in Holte et al.'s result for the non-erasing case), the premise $|\Sigma|\ge 3$ plays a crucial role in the proof. This raises the question whether the result generalizes to binary alphabets. The main contribution of this section is to show that Theorem~\ref{prop:Sigmage3} holds for binary alphabets if we limit the claim to patterns in which all variable blocks (except at the start and end of the pattern) contain at least 2 variables from each group, and all maximal terminal substrings are of length at least 3, while avoiding the shape $a^nb$ or $ab^n$ for $n\ge 2$. This result is formulated in Theorem~\ref{thm:mainResult}, which refers to a subclass of patterns defined as follows.

\medskip

\begin{definition}
Let $\mathcal{P}_{2,3}$ be the class of all relational patterns $(\hat{p},\hat{R})$ for which there exist 
     $n\in\mathbb{N}$, $\vec{x}_1, \vec{x}_{n+1} \in X^\ast$, $\vec{x}_2, \ldots, \vec{x}_n \in X^{\ge 2}$, and $\omega_j\in\Sigma^{\ge 3}$ for $1\le j\le n$, such that $\hat{p}$ is
     of the form 
    \[\hat{p}=\vec{x}_1 \omega_1 \vec{x}_2 \omega_2 \ldots \vec{x}_n \omega_n \vec{x}_{n+1}\,, \]
    and, for each $j\in\{2,\ldots,n\}$ and each group $[x]$, $\vec{x}_j$ contains at least 2 variables from~$[x]$.
\end{definition}

\medskip

\begin{theorem}\label{thm:mainResult}\!\!\!\!\footnote{The earlier conference version of this paper had a more general version of Theorem~\ref{thm:mainResult}, see~\cite[Theorem 2]{MousawiZ24}. However, the proof presented there has a gap, and we are currently not able to verify the stronger claim made in~\cite[Theorem 2]{MousawiZ24}.}
     Let $|\Sigma|=2$.
     Let $(p,R),(p',R')$ be any two relational patterns. Suppose  $(p,R)\in\mathcal{P}_{2,3}$.
     In addition, suppose $p$ has no terminal block of the form $a^n b$ or $a b^n$ where $a,b \in \Sigma$, $a\ne b$ and $n \geq 2$.     
     If $\mathcal{S}_{\varepsilon,2}(p,R)\subseteq\varepsilon L_{len}(p',R')\subseteq\varepsilon L_{len}(p,R)$, then $\varepsilon L_{len}(p,R)=\varepsilon L_{len}(p',R')$. In particular, $\mathcal{S}_{\varepsilon,2}(p,R)$ forms a telltale for $\varepsilon L_{len}(p,R)$ with respect to~$\varepsilon \mathcal{L}_{len}$.
\end{theorem}

\smallskip

The subsets of $\mathcal{S}_{\varepsilon,2}(p,R)$ that are exploited in the construction of our proof consist of a number of words linear in~$|p|$, which means that the languages in $\varepsilon \mathcal{L}_{len}$  have linearly-sized telltales.

This theorem is an immediate consequence of Theorems~\ref{thm:incongruous} and \ref{thm:congruous}, which will be proven in the subsequent subsections.

Theorems~\ref{thm:incongruous} and \ref{thm:congruous} offer a separate treatment of the claim of Theorem~\ref{thm:mainResult} for two cases, roughly speaking: (i) one in which the two patterns, whose languages are compared, have differences in the sequences of their terminal substrings, and (ii) one in which they don't. Formally, this is captured by the following definition.

\medskip

\begin{definition}
Let $n,n'\in\mathbb{N}$, $\vec{x}_1,\vec{x}'_1, \vec{x}_{n+1},\vec{x}'_{n'+1}  \in X^\ast$, $\vec{x}_1, \ldots, \vec{x}_n, \vec{x}_1, \ldots, \vec{x}_{n'} \in X^+$, and $\omega_j,\omega'_{j'}\in\Sigma^+$ for $1\le j\le n$, $1\le j'\le n'$. Consider the two patterns $p,p'$ with 
    \begin{eqnarray*}
    p&=&\vec{x}_1 \omega_1 \vec{x}_2 \omega_2 \ldots \vec{x}_n \omega_n \vec{x}_{n+1}\,,\\
    p'&=&\vec{x}'_1 \omega'_1 \vec{x}'_2 \omega'_2 \ldots \vec{x}'_{n'} \omega'_{n'} \vec{x}'_{n'+1}\,.
    \end{eqnarray*}
We call $p$ and $p'$ congruous, if $n=n'$ and $\omega_j=\omega'_{j'}$ for all $j$ with $1\le j\le n=n'$, and moreover, for $j\in\{1,n+1\}$, either both $\vec{x}_j$ and $\vec{x}'_j$ are empty or both are non-empty. Otherwise, $p$ and $p'$ are called incongruous.    
\end{definition}

\medskip

\begin{theorem}\label{thm:incongruous}
    Let $|\Sigma|=2$ and let $(p,R)$, $(p', R')$ be any two relational patterns, where $p$ and $p'$ are incongruous. Suppose $p$ has no terminal block of the form $a^n b$ or $ab^n$ where $a,b \in \Sigma$, $a \neq b$, and $n \geq 0$.\footnote{In Theorem~\ref{thm:mainResult}, the case $n\in\{0,1\}$ is impossible due to the premise $(p,R)\in \mathcal{P}_{2,3}$. In Theorem~\ref{thm:incongruous}, we need to mention the case $n\in\{0,1\}$ explicitly, thus obtaining ``$n\ge 0$'' instead of the condition ``$n\ge 2$'' in Theorem~\ref{thm:mainResult}.} If $\mathcal{S}_{\varepsilon,2}(p,R)\subseteq\varepsilon L_{len}(p',R')\subseteq\varepsilon L_{len}(p,R)$, then $\varepsilon L_{len}(p,R)=\varepsilon L_{len}(p',R')$.
\end{theorem}

\medskip

\begin{theorem}\label{thm:congruous}
    Let $|\Sigma|=2$ and let $(p,R)$, $(p',R')$ be any two relational patterns, where $p$ and $p'$ are congruous and $(p,R)\in\mathcal{P}_{2,3}$. If $\mathcal{S}_{\varepsilon,1}(p,R)\subseteq\varepsilon L_{len}(p',R')$, then $\varepsilon L_{len}(p,R)\subseteq\varepsilon L_{len}(p',R')$.
\end{theorem}

\medskip

\subsubsection{The Case of Incongruous Patterns}

This section is devoted to proving Theorem~\ref{thm:incongruous}, which states the following:
\begin{quote}
Let $|\Sigma|=2$ and let $(p,R)$, $(p',R')$ be any two relational patterns, where $p$ and $p'$ are incongruous. Suppose $p$ has no terminal block of the form $a^n b$ or $ab^n$ where $a,b \in \Sigma$, $a \neq b$, and $n \geq 0$. If $\mathcal{S}_{\varepsilon,2}(p,R)\subseteq\varepsilon L_{len}(p',R')\subseteq\varepsilon L_{len}(p,R)$, then $\varepsilon L_{len}(p,R)=\varepsilon L_{len}(p',R')$.
\end{quote}


 


The question whether Theorem~\ref{thm:incongruous}
remains valid when dropping the condition ``$p$ has no terminal block of the form $a^n b$ or $ab^n$ where $a,b \in \Sigma$, $a \neq b$, and $n \geq 0$'' remains open. Before we prove Theorem~\ref{thm:incongruous}, we first present a few special cases in which this theorem remains valid even if the pattern $p$ contains terminal blocks of the form $a^n b$ or $ab^n$ where $a,b \in \Sigma$, $a \neq b$, and $n \geq 0$.
These special cases are formulated in Lemmas \ref{lem:ex-|alphabet|=2-length-1} through \ref{lem:ex-|alphabet|=2-length-4}. Notably, these lemmas do not only concern testing \emph{equivalence}\/ of certain patterns via membership of a small number of short words, but demonstrate cases in which even \emph{inclusion}\/ can be tested this way. 

\medskip

\begin{lemma}\label{lem:ex-|alphabet|=2-length-1}
Let $\Sigma = \{a,b\}$, $a \neq b$, $n_i \in \mathbb{N}$, $\vec{x}_1, \vec{x}'_1, \vec{x}_3, \vec{x}'_3 \in X^\ast$, $\vec{x}_2, \vec{x}'_2 \in X^+$.
\begin{enumerate}
    \item Let $p = \vec{x}_1 a^{n_1} b^{n_2+n_3} \vec{x}_2 a \vec{x}_3$ and $p' = \vec{x}'_1 a^{n_1} b^{n_2} \vec{x}'_2 b^{n_3} a \vec{x}'_3$ be incongruous.  
    Fix $R,R'$. Then $\mathcal{S}_{\varepsilon,1}(p,R) \subseteq \varepsilon L_{len}(p',R')$ implies $\varepsilon L_{len}(p,R) \subseteq\varepsilon L_{len}(p',R')$.


    \item Let $p = \vec{x}_1 a b^{n_1} \vec{x}_2 b^{n_2} a^{n_3} \vec{x}_3$ and $p' = \vec{x}'_1 a \vec{x}'_2 b^{n_1+n_2} a^{n_3} \vec{x}'_3$ be incongruous.  
    Fix $R,R'$. Then $\mathcal{S}_{\varepsilon,1}(p',R') \subseteq \varepsilon L_{len}(p,R)$ implies $\varepsilon L_{len}(p',R') \subseteq\varepsilon L_{len}(p,R)$.
    
\end{enumerate}

\end{lemma}

\begin{proof}
    Let $c\in\{1,2,3\}$. If $x\in \mathit{Var}(p)$, then $d_c(x):=|[x]_R\cap \mathit{Var}(\vec{x}_c)|$ denotes the number of positions in $\vec{x}_c$ that correspond to variables in $[x]_R$. Likewise, for $x\in \mathit{Var}(p')$, we write $d'_c(x):=|[x]_{R'}\cap \mathit{Var}(\vec{x}'_c)|$.

    We only prove the first statement; the second one is proved analogously. So, let $(p,R)$ and $(p',R')$ be as in Statement 1. Note that this implies $n_3>0$, since otherwise $p$ and $p'$ would be congruous. Now suppose $\mathcal{S}_{\varepsilon,1}(p,R) \subseteq\varepsilon L_{len}(p',R')$. Without loss of generality, assume  $(p,R)$ is in equal-length normal form. First, we will establish the following claim.

    \emph{Claim.} Let $x\in \mathit{Var}(p)$, and let $\alpha,\beta\in\mathbb{N}$ satisfy $\alpha+\beta=d_2(x)$. 
    \begin{enumerate}
        \item If $n_1\ge 1$ and $n_2\ge 1$, then there is some $t>0$ and a family $(y_i)_{1\le i\le t}$ of variables in $p'$ such that the vector $\langle d_1(x),\alpha,d_3(x)+\beta\rangle$ is a linear combination of $\langle d'_1(y_i),d'_2(y_i),d'_3(y_i)\rangle$, $1\le i\le t$.
        
        \item If $n_2=0$, then there is some $t>0$ and a family $(y_i)_{1\le i\le t}$ of variables in $p'$ such that the vector $\langle d_1(x) + \alpha ,d_3(x)+ \beta \rangle$ is a linear combination of $\langle d'_1(y_i),d'_3(y_i)\rangle$, $1\le i\le t$, where $d'_2(y_i)=0$ for all~$i$.

        \item If $n_1=0$ and $n_2 \neq 0$, then there is some $t>0$, $\gamma \in \mathbb{N}$ with $\gamma \leq \alpha$, and a family $(y_i)_{1\le i\le t}$ of variables in $p'$ such that the vector $\langle d_1(x) + \gamma ,\alpha - \gamma,d_3(x)+\beta\rangle$ is a linear combination of $\langle d'_1(y_i),d'_2(y_i),d'_3(y_i)\rangle$, $1\le i\le t$.
    \end{enumerate}

    \emph{Proof of Claim 1 (for $n_1,n_2 \geq 1$.)} Since $(p,R)$ is in equal-length normal form, $\vec{x}_2$ contains a (possibly empty) substring of the form $z_1\ldots z_{d_2(x)}$, where $[x]\cap \mathit{Var}(\vec{x}_2)=\{z_1,\ldots, z_{d_2(x)}\}$. Let $\phi$ be the substitution that erases every variable not contained in $[x]$, and replaces variables $y\in [x]$ as follows: if $y\in [x]\cap (\mathit{Var}(\vec{x}_1) \cup \mathit{Var}(\vec{x}_3))$ or if $y\in\{z_1,\ldots, z_\alpha\}$, then $\phi(y)=b$; if $y\in\{z_{\alpha+1},\ldots, z_{d_2(x)}\}$, then $\phi(y)=a$. This substitution is valid for $R$; applied to $p$ it generates the word
    \[
    \phi(p)=b^{d_1(x)}\ a^{n_1}\ b^{n_2+n_3+\alpha}\ a^{\beta+1}\ b^{d_3(x)}\,.
    \]
    By definition of $\phi$, we have $\phi(p)\in\mathcal{S}_{\varepsilon,1}(p,R)$ and thus $\phi(p)\in\varepsilon L_{len}(p',R')$. Thus, there exists an $R'$-substitution $\phi'$ such that $\phi'(p')=\phi(p)$. Due to the shape of $p'$, this implies $\phi'(\vec{x}'_1)=b^{d_1(x)}$, $\phi'(\vec{x}'_2)=b^{\alpha}$,  and $\phi'(\vec{x}'_3)=a^{\beta}b^{d_3(x)}$. For $\phi'$ to be valid for $R'$, this entails the existence of a family $(y_i)_{1\le i\le t}$ of variables in $p'$ such that the vector $\langle d_1(x),\alpha,d_3(x)+\beta\rangle$ is a linear combination of the vectors $\langle d'_1(y_i),d'_2(y_i),d'_3(y_i)\rangle$, $1\le i\le t$. This completes the proof for Claim 1.
    
    \emph{Proof of Claim 2 (for $n_2 = 0$.)}
    Consider the following cases:
    \begin{itemize}
        \item $n_1 \neq 0,\ n_2 = 0$. In this case, the proof is identical to the proof for the case where $n_1, n_2 \geq 1$. Here, the string $\phi(p)$ has the following structure.
        \[
            \phi(p) = b^{d_1(x)}\ a^{n_1}\ b^{n_3+\alpha}\ a^{\beta+1}\ b^{d_3(x)}\,.
        \]
   
        \item $n_1 = n_2 = 0$. In this case, $p'$ has two variable blocks, i.e., \( p' = \vec{x}'_1\ b^{n_3}a\ \vec{x}'_3 \,, \) while \( p = \vec{x}_1\ b^{n_3}\ \vec{x}_2\ a\ \vec{x}_3 \,. \) 
        Now, let $\phi$ be the substitution that erases every variable not contained in $[x]$, and replaces variables $y\in [x]$ as follows: if $y\in [x]\cap \mathit{Var}(\vec{x}_3)$ or if $y\in\{z_1,\ldots, z_\alpha\}$, then $\phi(y)=b$; if $y\in ([x]\cap \mathit{Var}(\vec{x}_1))$ or if $y\in \{z_{\alpha+1},\ldots, z_{d_2(x)}\}$, then $\phi(y)=a$. This substitution is valid for $R$; applied to $p$ it generates the word
        \[
            \phi(p) = a^{d_1(x)}\ b^{n_3+\alpha}\ a^{\beta+1}\ b^{d_3(x)}\,.
        \]
        Again by definition of $\phi$, we have $\phi(p)\in\mathcal{S}_{\varepsilon,1}(p,R)$ and thus $\phi(p)\in\varepsilon L_{len}(p',R')$. Hence, there exists an $R'$-substitution $\phi'$ such that $\phi'(p')=\phi(p)$. Due to the shape of $p'$, this implies $\phi'(\vec{x}'_1)=a^{d_1(x)} b^{\alpha}$,  and $\phi'(\vec{x}'_3)=a^{\beta}b^{d_3(x)}$. 
        For $\phi'$ to be valid for $R'$, this entails the existence of a family $(y_i)_{1\le i\le t}$ of variables in $p'$ such that the vector $\langle d_1(x)+\alpha,d_3(x)+\beta\rangle$ is a linear combination of the vectors $\langle d'_1(y_i),d'_3(y_i)\rangle$, $1\le i\le t$, where $d'_2(y_i)=0$ for all~$i$.

         \end{itemize}
This completes the proof of Claim 2.
         
          \emph{Proof of Claim 3 (for $n_1 = 0,\ n_2 \neq 0$.)} In this case, let $\phi$ be the substitution that erases every variable not contained in $[x]$, and replaces variables $y\in [x]$ as follows: if $y\in [x]\cap \mathit{Var}(\vec{x}_3)$ or if $y\in\{z_1,\ldots, z_\alpha\}$, then $\phi(y)=b$; if $y\in ([x]\cap \mathit{Var}(\vec{x}_1))$ or if $y\in \{z_{\alpha+1},\ldots, z_{d_2(x)}\}$, then $\phi(y)=a$. This substitution is valid for $R$; applied to $p$ it generates the word
        \[
            \phi(p) = a^{d_1(x)}\  b^{n_2+n_3+\alpha}\ a^{\beta+1}\ b^{d_3(x)}\,.
        \]
        Again by definition of $\phi$, we have $\phi(p)\in\mathcal{S}_{\varepsilon,1}(p,R)$ and thus $\phi(p)\in\varepsilon L_{len}(p',R')$. Hence, there exists an $R'$-substitution $\phi'$ such that $\phi'(p')=\phi(p)$. Due to the shape of $p'$, this implies that there exists $\gamma \in \mathbb{N}$ such that $\gamma \leq \alpha$, $\phi'(\vec{x}'_1)=a^{d_1(x)} b^{\gamma}$, $\phi'(\vec{x}'_2)=b^{\alpha - \gamma}$, and $\phi'(\vec{x}'_3)=a^{\beta}b^{d_3(x)}$. 
        For $\phi'$ to be valid for $R'$, this entails the existence of a family $(y_i)_{1\le i\le t}$ of variables in $p'$ such that the vector $\langle d_1(x) + \gamma ,\alpha - \gamma,d_3(x)+\beta\rangle$ is a linear combination of $\langle d'_1(y_i),d'_2(y_i),d'_3(y_i)\rangle$, $1\le i\le t$.\hfill$\Box$\emph{(Claim.)}

\smallskip

    With this claim in hand, we can prove $\varepsilon L_{len}(p,R)\subseteq\varepsilon L_{len}(p',R')$. For the sake of brevity, let $n_1 \geq 1$ and $n_2 \geq 1$. The cases where $n_1 = 0$ or $n_2 = 0$ can be handled similarly. 
    
    Let $w\in\varepsilon L_{len}(p,R)$. We need to show that $w\in\varepsilon L_{len}(p',R')$.
    Clearly, there are $w_1,w_2,w_3\in\Sigma^*$ and variables $z_1, \dots, z_u$ in $p$ such that $w=w_1\ a^{n_1}\ b^{n_2+n_3}\ w_2\ a\ w_3$ and the vector $\langle |w_1|,|w_2|,|w_3|\rangle$ is a linear combination of the vectors $\langle d_1(z_j),d_2(z_j),d_3(z_j)\rangle$, $1\le j\le u$.
    In particular, let $\lambda_1, \ldots, \lambda_u$ be natural numbers such that 
    
    \[ \langle |w_1|,|w_2|,|w_3|\rangle = \sum_{j=1}^{u} \lambda_j \langle d_1(z_j),d_2(z_j),d_3(z_j)\rangle\,. \]
    
    Now, let $\alpha\ge 0$ be such that  $w_2 = b^{\alpha} s$  where $s$ either is the empty string or begins with the letter~$a$. Thus,
    \[ w=w_1\ a^{n_1}\ b^{n_2+n_3}\ w_2\ a\ w_3 = w_1\ a^{n_1}\ b^{n_2}\ b^{\alpha} \ b^{n_3}\ s\ a\ w_3\,. \]
    Clearly, since $s$ either is the empty string or begins with the letter $a$, we can fix $v \in \Sigma^\ast$ such that
    \[  w = w_1\ a^{n_1}\ b^{n_2}\ b^{\alpha} \ b^{n_3}\ a\ v\ w_3\,. \]

    Let $\beta = |w_2| - \alpha$. Since $|w_2| = \sum_j \lambda_j d_2(z_j)$, there exists an index $\newj$ and $0 \leq \alpha' \leq d_2(z_\newj)$ such that 
    \[ 
    \alpha = ( \sum_{j=1}^{\newj-1} \lambda_j d_2(z_j)) +  ((\lambda_\newj- k) d_2(z_\newj) + \alpha' )\,, 
    \]
    \[
    \beta = (\sum_{j=\newj+1}^{u} \lambda_j d_2(z_j)) +  ( kd_2(z_\newj) - \alpha' ) \,,
    \]
    for some $k$ with $0 < k \leq \lambda_\newj $.
    
    Let $\beta' = d_2(z_\newj)-\alpha'$. By the claim above, there is a family $(\overline{z}_i)_{1\le i\le \overline{t}}$ of variables in $p'$ such that the vector $\langle d_1(z_\newj),\alpha',d_3(z_\newj)+\beta' \rangle$ is a linear combination of the vectors $\langle d'_1(\overline{z}_i),d'_2(\overline{z}_i),d'_3(\overline{z}_i)\rangle$, $1\le i\le \overline{t}$. 
    In particular, there exist coefficients \( \overline{\lambda}_1, \ldots, \overline{\lambda}_{\overline{t}} \) in $\mathbb{N}$ such that 
    \[ \langle d_1(z_\newj), \alpha', d_3(z_\newj) + \beta' \rangle =  \sum_{i=1}^{\overline{t}} \overline{\lambda}_i \langle d'_1(\overline{z}_i),d'_2(\overline{z}_i),d'_3( \overline{z}_i) \rangle  \,. \]

    The above claim also implies that for any $j \in \{ 1, \dots, u \}$, there is a family $(\overline{z}^j_i)_{1\le i\le \eta_j}$ of variables in $p'$ and coefficients \( \overline{\lambda}^j_1, \ldots, \overline{\lambda}^j_{\eta_j} \) in $\mathbb{N}$ such that 
    \[ \langle d_1(z_j), d_2(z_j),d_3(z_j) \rangle =  \sum_{i=1}^{\eta_j} \overline{\lambda}^j_i \langle d'_1(\overline{z}^j_i),d'_2(\overline{z}^j_i),d'_3(\overline{z}^j_i) \rangle\,,  ~~~~~~~~~~ 1 \leq j < \newj \,,\]
    \[ \langle d_1(z_j),0,d_3(z_j)+d_2(z_j) \rangle =  \sum_{i=1}^{\eta_j} \overline{\lambda}^j_i \langle d'_1(\overline{z}^j_i),d'_2(\overline{z}^j_i),d'_3(\overline{z}^j_i) \rangle\,,  ~~~~~~~~~~ \newj+1 \leq j \leq u\,. \]

    Summing up, we obtain

    \begin{align*}
        \langle |w_1|,\alpha ,|w_3| + \beta \rangle & = \sum_{j=1}^{\newj-1} \lambda_j \langle d_1(z_j),d_2(z_j),d_3(z_j)\rangle + (\lambda_\newj - k ) \langle d_1(z_\newj),d_2(z_\newj),d_3(z_\newj) \rangle \\
        & ~~~~~~~  + \langle d_1(z_\newj), \alpha' ,d_3(z_\newj) + \beta' \rangle + (k-1) \langle d_1(z_\newj),0,d_3(z_\newj) + d_2(z_\newj) \rangle\\
        & ~~~~~~~  + \sum_{j=\newj+1}^{u} \lambda_j \langle d_1(z_j), 0, d_3(z_j) + d_2(z_j) \rangle \\
        & = \sum_{j=1}^{\newj-1} \sum_i \lambda_j \overline{\lambda}^j_i \langle d'_1(\overline{z}^j_i),d'_2(\overline{z}^j_i),d'_3(\overline{z}^j_i) \rangle \\
        & ~~~~~~~ + (\lambda_\newj-k) \sum_i \overline{\lambda}^\newj_i \langle d'_1(\overline{z}^\newj_i),d'_2(\overline{z}^\newj_i),d'_3(\overline{z}^\newj_i) \rangle \\
        & ~~~~~~~ + \sum_{i=1}^{\overline{t}} \overline{\lambda}_i \langle d'_1(\overline{z}_i),d'_2(\overline{z}_i),d'_3(\overline{z}_i)\rangle \\
        & ~~~~~~~ + (k-1) \sum_i \overline{\lambda}^\newj_i \langle d'_1(\overline{z}^\newj_i),d'_2(\overline{z}^\newj_i),d'_3(\overline{z}^\newj_i) \rangle \\
        & ~~~~~~~ + \sum_{j=\newj+1}^{u} \sum_i \lambda_j \overline{\lambda}^j_i \langle d'_1(\overline{z}^j_i),d'_2(\overline{z}^j_i),d'_3(\overline{z}^j_i) \rangle
    \end{align*}

    Thus, the vector $ \langle |w_1|,\alpha,|w_3|+\beta\rangle$ is a linear combination of the vectors $\langle d'_1(\overline{z}^j_i),d'_2(\overline{z}^j_i),d'_3(\overline{z}^j_i)\rangle$ and $\langle d'_1(\overline{z}_{i'}),d'_2(\overline{z}_{i'}),d'_3(\overline{z}_{i'})\rangle$ where $1\le j\le u$, $1\le i\le \eta_j$, and $1 \le i' \le \overline{t}$. Hence $w\in\varepsilon L_{len}(p',R')$, which completes the proof.
\end{proof}

Using similar techniques, we obtain the following lemma, whose proof is similar to that of Lemma~\ref{lem:ex-|alphabet|=2-length-1}.

\medskip

\begin{lemma}\label{lem:ex-|alphabet|=2-length-2}
Let $\Sigma = \{a,b\}$, $a \neq b$, $n_1, n_2, n_3 \in \mathbb{N}$, $n_2\ge 1$, $\vec{x}_1, \vec{x}'_1, \vec{x}_3, \vec{x}'_3 \in X^\ast$, $\vec{x}_2, \vec{x}'_2 \in X^+$. Suppose $p$ and $p'$ are incongruous patterns.
\begin{itemize}
    \item Let $p = \vec{x}_1 a^{n_1} b^{n_2} \vec{x}_2 a^{n_3} \vec{x}_3$ and $p' = \vec{x}'_1 a^{n_1} b^{n_2-1} \vec{x}'_2 b a^{n_3} \vec{x}'_3$. Fix $R,R'$. Then 
    $\mathcal{S}_{\varepsilon,1}(p,R) \subseteq \varepsilon L_{len}(p',R')$ implies $\varepsilon L_{len}(p,R) \subseteq\varepsilon L_{len}(p',R')$.
    
    \item Let $p = \vec{x}_1 a^{n_1} b \vec{x}_2 b^{n_2 -1} a^{n_3} \vec{x}_3$, and  $p' = \vec{x}'_1 a^{n_1} \vec{x}'_2 b^{n_2} a^{n_3} \vec{x}'_3$. Fix $R,R'$. Then
    $\mathcal{S}_{\varepsilon,1}(p',R') \subseteq \varepsilon L_{len}(p,R)$ implies $\varepsilon L_{len}(p',R') \subseteq \varepsilon L_{len} (p,R)$.
\end{itemize}
\end{lemma}

\begin{proof}
    Let $c\in\{1,2,3\}$. If $x\in \mathit{Var}(p)$, then $d_c(x):=|[x]_R\cap \mathit{Var}(\vec{x}_c)|$ denotes the number of positions in $\vec{x}_c$ that correspond to variables in $[x]_R$. Likewise, for $x\in \mathit{Var}(p')$, we write $d'_c(x):=|[x]_{R'}\cap \mathit{Var}(\vec{x}'_c)|$.

    We only prove the first statement; the second one is proved analogously. So, let $(p,R)$ and $(p',R')$ be as in Statement 1. Note that we exclude the case where $n_1 = n_3 = 0$ and $n_2 = 1$, since in that case $p$ and $p'$ are congruous.
    Now suppose $\mathcal{S}_{\varepsilon,1}(p,R) \subseteq\varepsilon L_{len}(p',R')$. Without loss of generality, assume  $(p,R)$ is in equal-length normal form. First, we will establish the following claim.

    \emph{Claim.} Let $x\in \mathit{Var}(p)$, and let $\alpha,\beta\in\mathbb{N}$ satisfy $\alpha+\beta=d_2(x)$. 

    \begin{enumerate}
        \item If $n_1 \ge 1$ and $n_3 \geq 1$, then there is some $t>0$ and a family $(y_i)_{1\le i\le t}$ of variables in $p'$ such that the vector $\langle d_1(x),\alpha,d_3(x)+\beta\rangle$ is a linear combination of $\langle d'_1(y_i),d'_2(y_i),d'_3(y_i)\rangle$, $1\le i\le t$.
        

        \item If $n_3=0$, then there is some $t>0$ and a family $(y_i)_{1\le i\le t}$ of variables in $p'$ such that the vector $\langle d_1(x), d_3(x) \rangle$ is a linear combination of $\langle d'_1(y_i),  d'_3(y_i)\rangle$, $1\le i\le t$, where $d'_2(y_i)=0$ for all $i$. (Note that in this case $p$ has only two variable blocks.)

        \item If $n_1=0$, $n_3 \geq 1$, and $n_2 \geq 2$, then there is some $t>0$, $\gamma \in \mathbb{N}$ with $\gamma \leq \alpha$, and a family $(y_i)_{1\le i\le t}$ of variables in $p'$ such that the vector $\langle d_1(x) + \gamma ,\alpha - \gamma,d_3(x)+\beta\rangle$ is a linear combination of $\langle d'_1(y_i),d'_2(y_i),d'_3(y_i)\rangle$, $1\le i\le t$.
        
        \item If $n_1=0$, $n_3 \geq 1$, and $n_2 = 1$, then there is some $t>0$ and a family $(y_i)_{1\le i\le t}$ of variables in $p'$ such that the vector $\langle d_1(x) + \alpha ,d_3(x)+ \beta \rangle$ is a linear combination of $\langle d'_1(y_i),d'_3(y_i)\rangle$, $1\le i\le t$, where $d'_2(y_i)=0$ for all~$i$.
    \end{enumerate}

    \emph{Proof of Claim 1 (for $n_1,n_3 \geq 1$.)} Since $(p,R)$ is in equal-length normal form, $\vec{x}_2$ contains a (possibly empty) substring of the form $z_1\ldots z_{d_2(x)}$, where $[x]\cap \mathit{Var}(\vec{x}_2)=\{z_1,\ldots, z_{d_2(x)}\}$. Let $\phi$ be the substitution that erases every variable not contained in $[x]$, and replaces variables $y\in [x]$ as follows: if $y\in [x]\cap (\mathit{Var}(\vec{x}_1) \cup \mathit{Var}(\vec{x}_3))$ or if $y\in\{z_1,\ldots, z_\alpha\}$, then $\phi(y)=b$; if $y\in\{z_{\alpha+1},\ldots, z_{d_2(x)}\}$, then $\phi(y)=a$. This substitution is valid for $R$; applied to $p$ it generates the word
    \[
    \phi(p)=b^{d_1(x)}\ a^{n_1}\ b^{n_2 + \alpha}\ a^{\beta+ n_3}\ b^{d_3(x)}\,.
    \]
    By definition of $\phi$, we have $\phi(p)\in\mathcal{S}_{\varepsilon,1}(p,R)$ and thus $\phi(p)\in\varepsilon L_{len}(p',R')$. Thus, there exists an $R'$-substitution $\phi'$ such that $\phi'(p')=\phi(p)$. Due to the shape of $p'$, this implies $\phi'(\vec{x}'_1)=b^{d_1(x)}$, $\phi'(\vec{x}'_2)=b^{\alpha}$,  and $\phi'(\vec{x}'_3)=a^{\beta}b^{d_3(x)}$. For $\phi'$ to be valid for $R'$, this entails the existence of a family $(y_i)_{1\le i\le t}$ of variables in $p'$ such that the vector $\langle d_1(x),\alpha,d_3(x)+\beta\rangle$ is a linear combination of the vectors $\langle d'_1(y_i),d'_2(y_i),d'_3(y_i)\rangle$, $1\le i\le t$. This completes the proof for Claim~1.

    \emph{Proof of Claim 2 (for $n_3 = 0$.)}
    In this case, $p$ has two variable blocks, i.e., w.l.o.g., \( p = \vec{x}_1\ a^{n_1} b^{n_2}\ \vec{x}_3 \,, \) while \( p' = \vec{x}'_1\ a^{n_1} b^{{n_2} -1}\ \vec{x}'_2\ b\ \vec{x}'_3 \,. \)

       Let $\phi$ be such that $\phi(y)=a$ for every $y\in [x]$ and $\phi(y)=\varepsilon$ for every $y\notin[x]$. This substitution is valid for $R$; applied to $p$ it generates the word
        \[
            \phi(p) = a^{d_1(x)}\ a^{n_1}\ b^{n_2}\ a^{d_3(x)}\,.
        \]
        By definition of $\phi$, we have $\phi(p)\in\mathcal{S}_{\varepsilon,1}(p,R)$ and thus $\phi(p)\in\varepsilon L_{len}(p',R')$. Thus, there exists an $R'$-substitution $\phi'$ such that $\phi'(p')=\phi(p)$. Due to the shape of $p'$, this implies $\phi'(\vec{x}'_1)=a^{d_1(x)}$, $\phi'(\vec{x}'_2)=\varepsilon $, and $\phi'(\vec{x}'_3)=a^{d_3(x)}$. For $\phi'$ to be valid for $R'$, this entails the existence of a family $(y_i)_{1\le i\le t}$ of variables in $p'$ such that the vector $\langle d_1(x), d_3(x) \rangle$ is a linear combination of the vectors $\langle d'_1(y_i), d'_3(y_i)\rangle$, $1\le i\le t$, where $d'_2(y_i)=0$ for all $i$. This completes the proof of Claim~2.
  

    \emph{Proof of Claim 3 (for $n_1=0$, $n_3 \geq 1$, and $n_2 \geq 2$.)}   
        In this case, $p$ and $p'$ are of the form \( p = \vec{x}_1\ b^{n_2}\ \vec{x}_2\ a^{n_3}\ \vec{x}_3 \,, \) and \( p' = \vec{x}'_1\ b^{{n_2}-1}\ \vec{x}'_2\ ba^{n_3}\ \vec{x}'_3 \,. \) 
        Now, let $\phi$ be the substitution that erases every variable not contained in $[x]$, and replaces variables $y\in [x]$ as follows: if $y\in [x]\cap \mathit{Var}(\vec{x}_3)$ or if $y\in\{z_1,\ldots, z_\alpha\}$, then $\phi(y)=b$; if $y\in ([x]\cap \mathit{Var}(\vec{x}_1))$ or if $y\in \{z_{\alpha+1},\ldots, z_{d_2(x)}\}$, then $\phi(y)=a$. This substitution is valid for $R$; applied to $p$ it generates the word
        \[
            \phi(p) = a^{d_1(x)}\ b^{n_2+\alpha}\ a^{\beta+n_3}\ b^{d_3(x)}\,.
        \]
        Again by definition of $\phi$, we have $\phi(p)\in\mathcal{S}_{\varepsilon,1}(p,R)$ and thus $\phi(p)\in\varepsilon L_{len}(p',R')$. Due to the shape of $p'$, this implies that there exists $\gamma \in \mathbb{N}$ such that $\gamma \leq \alpha$, $\phi'(\vec{x}'_1)=a^{d_1(x)} b^{\gamma}$, $\phi'(\vec{x}'_2)=b^{\alpha - \gamma}$, and $\phi'(\vec{x}'_3)=a^{\beta}b^{d_3(x)}$. 
        For $\phi'$ to be valid for $R'$, this entails the existence of a family $(y_i)_{1\le i\le t}$ of variables in $p'$ such that the vector $\langle d_1(x) + \gamma ,\alpha - \gamma,d_3(x)+\beta\rangle$ is a linear combination of $\langle d'_1(y_i),d'_2(y_i),d'_3(y_i)\rangle$, $1\le i\le t$.
This completes the proof of Claim~3.

          \emph{Proof of Claim 4 ($n_1=0$, $n_3 \geq 1$, and $n_2 = 1$.)} In this case, $p$ and $p'$ are of the form \( p = \vec{x}_1\ b\ \vec{x}_2\ a^{n_3}\ \vec{x}_3 \) and \( p' = \vec{x}'_1\ ba^{n_3}\ \vec{x}'_3 \,. \) 
          Now, let $\phi$ be the substitution that erases every variable not contained in $[x]$, and replaces variables $y\in [x]$ as follows: if $y\in [x]\cap \mathit{Var}(\vec{x}_3)$ or if $y\in\{z_1,\ldots, z_\alpha\}$, then $\phi(y)=b$; if $y\in ([x]\cap \mathit{Var}(\vec{x}_1))$ or if $y\in \{z_{\alpha+1},\ldots, z_{d_2(x)}\}$, then $\phi(y)=a$. This substitution is valid for $R$; applied to $p$ it generates the word
        \[
            \phi(p) = a^{d_1(x)}\  b^{1 + \alpha}\ a^{\beta + n_3}\ b^{d_3(x)}\,.
        \]
        Again by definition of $\phi$, we have $\phi(p) \in \mathcal{S}_{\varepsilon,1}(p,R)$ and thus $\phi(p)\in\varepsilon L_{len}(p',R')$. Hence, there exists an $R'$-substitution $\phi'$ such that $\phi'(p')=\phi(p)$. Due to the shape of $p'$, this implies $\phi'(\vec{x}'_1)=a^{d_1(x)} b^{\alpha}$,  and $\phi'(\vec{x}'_3)=a^{\beta}b^{d_3(x)}$. 
        For $\phi'$ to be valid for $R'$, this entails the existence of a family $(y_i)_{1\le i\le t}$ of variables in $p'$ such that the vector $\langle d_1(x)+\alpha,d_3(x)+\beta\rangle$ is a linear combination of the vectors $\langle d'_1(y_i),d'_3(y_i)\rangle$, $1\le i\le t$, where $d'_2(y_i)=0$ for all~$i$.  \hfill$\Box$\emph{(Claim.)}

\smallskip

    With this claim in hand, we can prove $\varepsilon L_{len}(p,R)\subseteq\varepsilon L_{len}(p',R')$. For the sake of brevity, let $n_1 \geq 1$ and $n_3 \geq 1$. The cases where $n_1 = 0$ or $n_3 = 0$ can be handled similarly.

    Let $w\in\varepsilon L_{len}(p,R)$. We need to show that $w\in\varepsilon L_{len}(p',R')$.
    Clearly, there are $w_1,w_2,w_3\in\Sigma^\ast$ and variables $z_1, \dots, z_u$ in $p$ such that $w=w_1\ a^{n_1}\ b^{n_2}\ w_2\ a^{n_3}\ w_3$ and the vector $\langle |w_1|,|w_2|,|w_3|\rangle$ is a linear combination of the vectors $\langle d_1(z_j),d_2(z_j),d_3(z_j)\rangle$, $1\le j\le u$.
    In particular, let $\lambda_1, \ldots, \lambda_u$ be natural numbers such that 
    \[ \langle |w_1|,|w_2|,|w_3|\rangle = \sum_{j=1}^{u} \lambda_j \langle d_1(z_j),d_2(z_j),d_3(z_j)\rangle\,. \]  
    Now, let $w_2 = s\ s'$, where $s' \in \{a\}^\ast$ and $s$ either is the empty string or ends with the letter~$b$. Thus,
    \[ w=w_1\ a^{n_1}\ b^{n_2}\ w_2\ a^{n_3}\ w_3 = w_1\ a^{n_1}\ b^{n_2 -1}\ b\ s\ s'\ a^{n_3}\ w_3\,. \]    
    Clearly, since $s$ either is the empty string or ends with the letter $b$, we can fix $v \in \Sigma^\ast$ such that $|v|=|s|$ and
    \[  w = w_1\ a^{n_1}\ b^{n_2 -1}\ v\ b\ a^{n_3}\ s'\ w_3\,. \]
    Let $\alpha = |v| = |s|$ and $\beta = |w_2| - \alpha = |s'|$. Since $|w_2| = \sum_j \lambda_j d_2(z_j)$, there exists an index $\newj$ and $0 \leq \alpha' \leq d_2(z_\newj)$ such that 
    \[ 
    \alpha = \sum_{j=1}^{\newj-1} \lambda_j d_2(z_j) +  ((\lambda_\newj- k) d_2(z_\newj) + \alpha' )\,, 
    \]
    \[
    \beta = \sum_{j=\newj+1}^{u} \lambda_j d_2(z_j) +  ( kd_2(z_\newj) - \alpha' ) \,,
    \]
    for some $k$ with $0 < k \leq \lambda_\newj $.  
    Let $\beta' = d_2(z_\newj)-\alpha'$. We now invoke Claim 1 with $(\alpha',\beta')$ in place of $(\alpha,\beta)$, so that the remainder of the proof is identical to the corresponding part of the proof of Lemma~\ref{lem:ex-|alphabet|=2-length-1}.
\end{proof}

The proof of Lemma~\ref{lem:ex-|alphabet|=2} is fairly simple.

\medskip

\begin{lemma}\label{lem:ex-|alphabet|=2}
Let $\Sigma = \{a,b\}$ and $a \neq b$. Let $p = \vec{x}_1 a^{n_1} b^{n_2} \vec{x}_2 a^{n_3} \vec{x}_3$ and $p' = \vec{x}'_1 a^{n_1} \vec{x}'_2 b^{n_2} a^{n_3} \vec{x}'_3$, where $n_1,n_2,n_3 \in \mathbb{N}$, $\vec{x}_1, \vec{x}'_1, \vec{x}_3, \vec{x}'_3 \in X^\ast$ and $\vec{x}_2, \vec{x}'_2 \in X^+$. Suppose $p, p'$ are incongruous patterns and let $R$ and $R'$ be arbitrary relations over $X$.  Then:
\begin{enumerate}
    \item $\mathcal{S}_{\varepsilon, 2} (p, R) \subseteq \varepsilon L_{len}(p',R')$ implies $\varepsilon L_{len}(p,R) \subseteq \varepsilon L_{len}(p',R')$.
    \item $\mathcal{S}_{\varepsilon, 2} (p', R') \subseteq \varepsilon L_{len}(p,R)$ implies $\varepsilon L_{len}(p',R') \subseteq \varepsilon L_{len}(p,R)$.
\end{enumerate}
\end{lemma}

\begin{proof} If $n_1 = 1$ or $n_3 = 1$, both statements follow from Lemma~\ref{lem:ex-|alphabet|=2-length-1}. 
    If $n_2= 0$, since $p$ and $p'$ are incongruous, we have that exactly one of $\vec{x}_1,\vec{x}'_1$ is empty or exactly one of $\vec{x}_3,\vec{x}'_3$ is empty (where the ``or'' is not exclusive). By symmetry, it suffices to consider the cases (i) $\vec{x}_1=\vec{x}_3=\varepsilon$ and $\vec{x}'_1\ne\varepsilon\ne\vec{x}'_3$, (ii) $\vec{x}_1=\vec{x}'_3=\varepsilon$ and $\vec{x}'_1\ne\varepsilon\ne\vec{x}_3$, and (iii) $\vec{x}_1=\varepsilon\ne\vec{x}_3$ and $\vec{x}'_1\ne\varepsilon\ne\vec{x}'_3$. In each of these cases, $\mathcal{S}_{\varepsilon, 2} (p', R') \not\subseteq \varepsilon L_{len}(p,R)$, proving Statement 2. In Case (i), $\mathcal{S}_{\varepsilon, 2} (p, R) \subseteq \varepsilon L_{len}(p',R')$ implies $a^{n_1}b^ma^{n_3}\in L_{len}(p',R')$, where $m$ is any arbitrary group size in $R$. Hence $R'$ must have variable groups fully contained in $\vec{x}'_2$, whose sizes combine linearly to $m$. Consequently, since this applies to any group size $m$ in $R$, we get $\varepsilon L_{len}(p,R) \subseteq \varepsilon L_{len}(p',R')$, i.e., Statement 1 holds. In Case (ii), Statement 1 holds since $\mathcal{S}_{\varepsilon, 2} (p, R) \not\subseteq \varepsilon L_{len}(p',R')$. In Case (iii), consider any group $[x]$ of size $m$ in $R$, with $m_2$ members in $\vec{x}_2$ and $m_3$ members in $\vec{x}_3$. Then $\mathcal{S}_{\varepsilon, 2} (p, R) \subseteq \varepsilon L_{len}(p',R')$ implies $a^{n_1}b^{m_2}a^{n_3}b^{m_3}\in L_{len}(p',R')$. Therefore $\langle m_2,m_3\rangle$ must be a linear combination of vectors $\langle m'_{2,i},m'_{3,i}\rangle$, such that, for each $i$, the relation $R'$ has a group with $m'_{2,i}$ members in $\vec{x}'_2$ and $m'_{3,i}$ members in $\vec{x}'_3$. Since this holds for any group $[x]$ in $R$, we obtain $\varepsilon L_{len}(p,R) \subseteq \varepsilon L_{len}(p',R')$, i.e., Statement 1 holds.
    
    Hence, from now on, suppose $n_1,n_3\ne 1$ and $n_2\ge 1$.

    If $n_3=0$, then wlog suppose $\vec{x}_3 = \varepsilon$, so that we have $p = \vec{x}_1 a^{n_1} b^{n_2} \vec{x}_2$ and $p' = \vec{x}'_1 a^{n_1} \vec{x}'_2 b^{n_2}  \vec{x}'_3$. Consider the following cases:
    \begin{itemize}
        \item $n_1\ge 2$ and $n_2\ge 2$. To verify Statement 1, suppose $\mathcal{S}_{\varepsilon, 2} (p, R) \subseteq \varepsilon L_{len}(p',R')$. Now consider any variable $x$ in $p$ and define  the $\ell_1$-substitution $\theta_x$ that erases all variables outside $[x]$, replaces all variables in $[x]$ that occur in $\vec{x}_1$ with $b$, and replaces all variables in $[x]$ that occur in $\vec{x}_2$ with $a$. Applied to $p$, it generates the string $\theta_x(p)=b^{d_1(x)}a^{n_1}b^{n_2}a^{d_2(x)}$, where $d_i(x)$ is the number of occurrences of variables in $[x]$ in the string $\vec{x}_i$, for $i\in\{1,2\}$. Obviously, \( \theta_x(p) \in \mathcal{S}_{\varepsilon, 1} (p, R) \subseteq \mathcal{S}_{\varepsilon, 2} (p, R) \), and thus $\theta_x(p) \in \varepsilon L_{len}(p',R')$ for each variable $x$ in~$p$. 

        Clearly, any $R'$-substitution $\theta'_x$ with $\theta'_x(p') = \theta_x(p)$ must satisfy $\theta'_x(\vec{x}'_2) = \varepsilon$. This implies that there are variables $y_1, \dots, y_m$ in $p'$ such that $d_1(x) = \lambda_1 d'_1(y_1) + \dots + \lambda_m d'_1(y_m)$, $d_2(x) = \lambda_1 d'_3(y_1) + \dots + \lambda_m d'_3(y_m)$, and $d'_2(y_1) = \dots = d'_2(y_m) = 0$, where $\lambda_1, \dots, \lambda_m \geq 1$ are suitably chosen coefficients and $d'_i(y)$ is the number of occurrences of variables in $[y]$ in the string $\vec{x}'_i$. Since such coefficients $\lambda_1, \dots, \lambda_m \geq 1$ exist for any variable $x$ in $p$, we obtain $\varepsilon\mathcal{L} (\vec{x}_1 a^{n_1} b^{n_2} \vec{x}_2,R)\subseteq\varepsilon\mathcal{L}(\vec{x}'_1 a^{n_1} \vec{x}'_2 b^{n_2}  \vec{x}'_3,R')$.
        
        Statement 2 follows, since \( \mathcal{S}_{\varepsilon, 2} (p', R') \not \subseteq \varepsilon L_{len}(p,R) \). To see the latter, we argue as follows. If there is a free variable $y\in \mathit{Var}(\vec{x}'_2)$, then $\varphi(y)=ba$ and $\varphi(z)=\varepsilon$ for $z\ne y$ yields an $R'$-substitution $\varphi$ such that $\varphi(p')\in\mathcal{S}_{\varepsilon, 2} (p', R')\setminus \varepsilon L_{len}(p,R)$. If no variable in $\mathit{Var}(\vec{x}'_2)$ is free, then let $y$ be the rightmost variable in $\mathit{Var}(\vec{x}'_2)$. Let $\varphi(y)=ba$, $\varphi(z)=aa$ for all $z\in[y]\cap(\mathit{Var}(\vec{x}'_1)\cup \mathit{Var}(\vec{x}'_2))$,  $\varphi(z)=aa$ for all $z\in[y]\cap \mathit{Var}(\vec{x}'_3)$, and $\varphi(z)=\varepsilon$ otherwise. It is not hard to see that $\varphi(p')\in\mathcal{S}_{\varepsilon, 2} (p', R')\setminus \varepsilon L_{len}(p,R)$.

        \item $n_1\ge 2$ and $n_2=1$. Now both statements follow from Lemma~\ref{lem:ex-|alphabet|=2-length-2}.
        
        \item $n_1=0$. In this case, we can assume wlog that $\vec{x}'_1=\varepsilon$, so that $p = \vec{x}_1 b^{n_2} \vec{x}_2$ and $p' =  \vec{x}'_2 b^{n_2}  \vec{x}'_3$. 
        Now both statements are easily verified as in the case $n_2=0$ discussed above.
    \end{itemize}
This covers all cases for $n_3=0$. By symmetry, one can handle the case $n_1=0$. Hence, from now on, we assume $n_1,n_3\ge 2$ and $n_2\ge 1$. We consider two cases:
\begin{itemize}
    \item $n_2=1$. Here both statements follow from Lemma~\ref{lem:ex-|alphabet|=2-length-2}.
    \item $n_2\ge 2$. Then, to verify the claim, one can use similar substitutions as when verifying statement 2 in case $n_1\ge 2$, $n_2\ge 2$, and $n_3=0$.
\end{itemize}
\end{proof}

Finally, we present one more special case in which Theorem~\ref{thm:incongruous} remains valid even if the pattern $p$ contains terminal blocks of the form $a^n b$ or $ab^n$ where $a,b \in \Sigma$, $a \neq b$, and $n \geq 0$.

\medskip

\begin{lemma}\label{lem:ex-|alphabet|=2-length-4}
    Suppose $p$ and $p'$ are incongruous patterns of the form
    \begin{eqnarray*}
    p &=& \vec{x}_1 a^{n_1}b^{m_1} \vec{x}_2 a^{n_2}b^{m_2} \ldots \vec{x}_{t} a^{n_t}b^{m_t} \vec{x}_{t+1}\mbox{ and}\\
    p' &=& \vec{x}'_1 a^{n_1} \vec{x}'_2 b^{m_1}a^{n_2} \ldots \vec{x}'_{t} b^{m_{t-1}} a^{n_t} \vec{x}'_{t+1} b^{m_t} \vec{x}'_{t+2}\,,
    \end{eqnarray*}
    where $n_i, m_j \in \mathbb{N} \setminus \{0\}$, $\vec{x}_k, \vec{x}'_{k'} \in X^+$ and $\vec{x}_1, \vec{x}'_1, \vec{x}_{t+1}, \vec{x}'_{t+2} \in X^\ast$, for all $i,j \in \{1, \ldots, t\}$, all $k\in \{2, \ldots, t\}$, and all $k'\in \{2, \ldots, t+1\}$. Suppose the following conditions hold:
\begin{itemize}
    \item  for all $i\in\{2,\ldots,t\}$: $n_i=1$ or ($m_{i-1}=1$ and $m_i=1$).
    \item for all $i\in\{1,\ldots,t-1\}$: $m_i=1$ or ($n_{i}=1$ and $n_{i+1}=1$).
    \item if $t=1$, then $n_1=1$ or $m_1=1$.
    \end{itemize}
   Then, for any two relations $R$ and $R'$ on $X$, $\mathcal{S}_{\varepsilon, 1} (p,R) \subseteq \varepsilon L_{len}(p',R')$ implies $\varepsilon L_{len}(p,R) \subseteq \varepsilon L_{len}(p',R')$.
\end{lemma}



\begin{proof}


    Let $j\in\{1,\ldots,t+1\}$. If $x\in \mathit{Var}(p)$, then $d_j(x):=|[x]_R\cap \mathit{Var}(\vec{x}_j)|$ denotes the number of positions in $\vec{x}_j$ that correspond to variables in $[x]_R$. Now, for any $x\in \mathit{Var}(p)$, consider all $\ell_1$-substitutions $\theta$ valid for $R$ that erase variables outside $[x]$ and replace any occurrence of a variable from $[x]$ in $\vec{x}_j$: (i) with $b$, if $j=1$, (ii) with $a$, if $j=t+1$, and (iii) with a single letter so that $\theta(\vec{x}_j)$ is of the form $b^{\beta_j} a^{\alpha_j}$ for some $\beta_j,\alpha_j$ with $\beta_j+\alpha_j=d_j(x)$. These substitutions, applied to $p$, generate words of the form
    \begin{equation}\label{eqn:string}
        b^{d_1(x)} \ a^{n_1}b^{m_1}\ b^{\beta_2}a^{\alpha_2}\ a^{n_2}b^{m_2}\ b^{\beta_3}a^{\alpha_3}\ \dots\ b^{\beta_t}a^{\alpha_t}\ a^{n_t}b^{m_t}\ a^{d_{t+1}(x)}\,,
    \end{equation}
    with $\beta_j+\alpha_j=d_j(x)$. All of these strings belong to $\mathcal{S}_{\varepsilon, 2} (p,R)$, and thus, by the premise of the lemma, to $\varepsilon L_{len}(p',R')$. Therefore, for each word $w$ of the form (\ref{eqn:string}), there is an $R'$-substitution $\theta'$ such that $\theta'(p')=w$.

Now, we want to show that for every \( h \in \{ 1, \dots, t+2 \}\) it holds that
    \begin{equation}\label{eqn:str==}
        \begin{split}
            \theta'(\vec{x}'_1 a^{n_1} \vec{x}'_2 b^{m_1}a^{n_2} \vec{x}'_3 \dots a^{n_{h-2}}\vec{x}'_{h-1} b^{m_{h-2}})\\ 
            = b^{d_1(x)} \ a^{n_1}b^{m_1}\ b^{\beta_2}a^{\alpha_2}\ a^{n_2}b^{m_2}\ \dots\ b^{m_{h-2}} b^{\beta_{h-1}}
        \end{split}
    \end{equation}
    Assume that $\theta'(\vec{x}'_1)$ starts with $b^{d_1(x)} a$. Then, since $\theta'(\vec{x}'_1 a^{n_1} \vec{x}'_2) = \theta'(\vec{x}'_1) a^{n_1} \theta'(\vec{x}'_2)$, we obtain that $\theta'(\vec{x}'_1 a^{n_1} \vec{x}'_2)$ starts with $b^{d_1(x)}\ a^{n_1}b^{m_1}\ b^{\beta_2}a^{\alpha_2}\ a^{n_2}$. Iteratively, we can conclude that \( \theta'(\vec{x}'_1 a^{n_1} \vec{x}'_2 b^{m_1} a^{n_2} \vec{x}'_3 \dots \vec{x}'_{t} b^{m_{t-1}} a^{n_t} \vec{x}'_{t+1} ) \) will start with \( b^{d_1(x)}\ a^{n_1}b^{m_1}\ b^{\beta_2}a^{\alpha_2}\ a^{n_2} b^{m_2}\ b^{\beta_3}a^{\alpha_3}\ a^{n_3} b^{m_3}\ \dots\ b^{\beta_t}a^{\alpha_t}\ a^{n_t}b^{m_t}\ a^{d_{t+1}(x)} \). However, this yields 
    \[ \theta'(p') = \theta'(\vec{x}'_1 a^{n_1} \vec{x}'_2 b^{m_1} a^{n_2} \vec{x}'_3 \dots \vec{x}'_{t} b^{m_{t-1}} a^{n_t} \vec{x}'_{t+1} b^{m_t} \vec{x}'_{t+2} ) \neq \theta(p) \] which is a contradiction. Therefore, $\theta'(\vec{x}'_1)$ is a prefix of $b^{d_1(x)}$.  Iteratively, we conclude that for every \( h \in \{ 1, \dots, t+1 \} \)    
    \begin{equation}\label{eqn:strCntnRight}
    \begin{split}
    \theta'(\vec{x}'_1\ a^{n_1}\ \vec{x}'_2\ b^{m_1} a^{n_2}\ \vec{x_3} \dots b^{m_{h-2}}a^{n_{h-1}}\ \vec{x}'_h)
    \mbox{ is a prefix of } \\ b^{d_1(x)} \ a^{n_1}b^{m_1}\ b^{\beta_2}a^{\alpha_2}\ a^{n_2}b^{m_2}\ b^{\beta_3}a^{\alpha_3}\ \dots\ b^{\beta_{h-1}}a^{\alpha_{h-1}}\ a^{n_{h-1}}b^{m_{h-1}} b^{\beta_h}
    \end{split}
    \end{equation}
    
    Similarly, we can conclude that for every \( h \in \{ 2, \dots, t+2 \} \) 
    \begin{equation}\label{eqn:strCntnLeft}
    \begin{split}
    \theta'(\vec{x}'_h b^{m_{h-1}}a^{n_h} \ldots \vec{x}'_{t} b^{m_{t-1}} a^{n_t} \vec{x}'_{t+1} b^{m_t} \vec{x}'_{t+2}) \mbox{ is a suffix of} \\
    a^{\alpha_{h-1}}\ a^{n_{h-1}} b^{m_{h-1}} b^{\beta_h}a^{\alpha_h} a^{n_h}  \dots\ b^{\beta_t}a^{\alpha_t}\ a^{n_t}b^{m_t}\ a^{d_{t+1}(x)}\,,
    \end{split}
    \end{equation}
    This implies, for every $h\in\{2,\ldots,t+1\}$,
    \begin{equation}\label{eqn:str=}
    a^{n_{h-1}}\theta'(\vec{x}'_h) b^{m_{h-1}} = \underbrace{\theta(\vec{x}_{h-1})[\beta_{h-1}+1:d_{h-1}(x)]}_{=a^{\alpha_{h-1}}}\ a^{n_{h-1}}b^{m_{h-1}}\ \underbrace{\theta(\vec{x}_h)[1:\beta_h]}_{=b^{\beta_h}}\,,
    \end{equation}
    where $\alpha_1=0$.
    Since $n_{h} = 1$ or $m_{h} = 1$ for $h \in \{ 1, \dots, t \}$, using Equation \ref{eqn:str=} and a similar argument as in the proof of Lemma~\ref{lem:ex-|alphabet|=2-length-1} and Lemma~\ref{lem:ex-|alphabet|=2-length-2}, one can conclude that $\mathcal{S}_{\varepsilon, 1} (p,R) \subseteq \varepsilon L_{len}(p',R')$ implies $\varepsilon L_{len}(p,R) \subseteq \varepsilon L_{len}(p',R')$. 
\end{proof}

\medskip

Lemma~\ref{lem:ex-|alphabet|=2-length-4} motivates the following definition, which will be crucial for the proof of Theorem~\ref{thm:incongruous}. 

\medskip

\begin{definition}\label{def:conjugate-pat}
  Let $t\ge 1$. Let
    \[ \pi = a^{n_1}b^{m_1} \vec{x}_1 a^{n_2}b^{m_2} \ldots \vec{x}_{t-1} a^{n_t}b^{m_t}\mbox{ and }\pi' = a^{n_1} \vec{x}'_1 b^{m_1}a^{n_2} \ldots \vec{x}'_{t-1} b^{m_{t-1}} a^{n_t} \vec{x}'_{t} b^{m_t} \,,\]
    or 
    \[ t>1 \mbox{ and }\pi = a^{n_1}b^{m_1} \vec{x}_1 a^{n_2}b^{m_2} \ldots \vec{x}_{t-1} a^{n_t}\mbox{ and }\pi' = a^{n_1} \vec{x}'_1 b^{m_1}a^{n_2} \ldots \vec{x}'_{t-1} b^{m_{t-1}} a^{n_t}   \,,\]
where $n_i, m_j \in \mathbb{N} \setminus \{0\}$,  and  $\vec{x}_k, \vec{x}'_{k'} \in X^+$, for all $i,j,k' \in \{1, \ldots, t\}$ and all $k\in \{1, \ldots, t-1\}$. 

    Then $\pi$ and $\pi'$ are called telltale conjugates
    provided that  
    the following conditions hold:
\begin{itemize}
    \item  for all $i\in\{2,\ldots,t\}$: $n_i=1$ or ($m_{i-1}=1$ and ($m_i=1$ if $m_i$ exists\footnote{Note that $m_t$ might not exist if $t>1$.})).
    \item for all $i\in\{1,\ldots,t-1\}$: $m_i=1$ or ($n_{i}=1$ and $n_{i+1}=1$).
    \item if $t=1$, then $n_1=1$ or $m_1=1$.
\end{itemize}
This definition is symmetric in the sense that $\pi$ and $\pi'$ are telltale conjugates if and only if $\pi'$ and $\pi$ are telltale conjugates; moreover, the roles of $a$ and $b$ are interchangeable.
\end{definition}



\medskip

The main technical ingredient of the proof of Theorem~\ref{thm:incongruous} is the following lemma, which can be proven with a lengthy case analysis. The proof is provided in the appendix.

\medskip

\begin{lemma}\label{lem:ex-|alphabet|=2-length-5}
    Let $(p,R)$ and $(p',R')$ be arbitrary patterns over $\Sigma = \{a,b\}$, where $p$ and $p'$ are incongruous. Suppose there do not exist telltale conjugates $\pi$, $\pi'$ such that 
    \begin{itemize}
        \item $\pi$ is a substring of $p$ that is not preceded or succeeded by terminal symbols in $p$, and
        \item $\pi'$ is a substring of $ p'$  that is not preceded or succeeded by terminal symbols in $p'$.
    \end{itemize}  
    Let  $\mathcal{S}_{\varepsilon,2}(p,R) \subseteq \varepsilon L_{len}(p',R') \subseteq \varepsilon L_{len}(p,R)$. Then $\varepsilon L_{len}(p,R)= \varepsilon L_{len}(p',R')$. 
\end{lemma}

\medskip

Now, we can put everything together to prove Theorem \ref{thm:incongruous}, which states the following:

\begin{quote}
    Let $|\Sigma|=2$ and let $(p,R)$, $(p', R')$ be any two relational patterns, where $p$ and $p'$ are incongruous. Suppose $p$ has no terminal block of the form $a^n b$ or $ab^n$ where $a,b \in \Sigma$, $a \neq b$, and $n \geq 0$.
     If $\mathcal{S}_{\varepsilon,2}(p,R)\subseteq\varepsilon L_{len}(p',R')\subseteq\varepsilon L_{len}(p,R)$, then $\varepsilon L_{len}(p,R)=\varepsilon L_{len}(p',R')$.
\end{quote}


\begin{proof}
    By Lemma~\ref{lem:ex-|alphabet|=2-length-5}, it suffices to prove that $p$ and $p'$ do not have telltale conjugate subpatterns that are not preceded or succeeded by terminal symbols.

    If $p$ and $p'$ did have such subpatterns, then at least one of the following two conditions would have to hold, by Definition~\ref{def:conjugate-pat}. 
    \begin{enumerate}
        \item $p$ has a terminal block of the form $\sigma^n\overline{\sigma}$ or $\sigma\overline{\sigma}^n$, where $n\ge 1$ and $\{\sigma,\overline{\sigma}\}=\{a,b\}=\Sigma$.
        \item $p$ is of the form $\vec{x_1}\sigma^{n_1}\vec{x_2}\overline{\sigma}^{m_1}\vec{x_3}$, where $n_1=1$ or $m_1=1$, and $\{\sigma,\overline{\sigma}\}=\{a,b\}=\Sigma$.
    \end{enumerate}
    Each of these conditions is impossible, due to the premise of the theorem. Note that the second condition violates the premise with $n=0$.

    This completes the proof.
\end{proof}

\subsubsection{The Case of Congruous Patterns}

This subsection deals with the proof of Theorem~\ref{thm:congruous}, that is, its goal is to establish the following.
\begin{quote}
    Let $|\Sigma|=2$ and let $(p,R)$, $(p',R')$ be any two relational patterns, where $p$ and $p'$ are congruous and $(p,R)\in\mathcal{P}_{2,3}$. If $\mathcal{S}_{\varepsilon,1}(p,R)\subseteq\varepsilon L_{len}(p',R')$, then $\varepsilon L_{len}(p,R)\subseteq\varepsilon L_{len}(p',R')$.
\end{quote}

Note that this statement is stronger than the claim that linear-size positive characteristic sets exist for a specific class of relational patterns. In particular, it implies that one can efficiently (in polynomial time) test the inclusion of two equal-length relational pattern languages, if the underlying patterns are congruous and have variable blocks intersecting each group at least twice, as well as terminal blocks of length at least 3. Only $\ell_1$-substitutions need to be invoked for this test. The downside of obtaining this stronger statement is that it is limited to a special sub-class of patterns.

\begin{proof}
Suppose that $\kappa$ (resp.\ $\kappa'$) is the number of groups in $p$ (resp.\ $p'$), so that in normal-form representation of $(p,R)$ and $(p',R')$ (see Definition~\ref{equal-length normal form}), we have  
    \[p = \vec{x}_1(1) \dots \vec{x}_1(\kappa) \omega_1\vec{x}_2(1) \dots \vec{x}_2(\kappa) \omega_2  \dots \omega_{n}\vec{x}_{n+1}(1) \dots \vec{x}_{n+1}(\kappa)\,,\]
    \[p' = \vec{x}'_1(1) \dots \vec{x}'_1(\kappa') \omega_1\vec{x}'_2(1) \dots \vec{x}'_2(\kappa') \omega_2  \dots \omega_{n}\vec{x}'_{n+1}(1) \dots \vec{x}'_{n+1}(\kappa')\,.\]
    For $1\le j\le n+1$, the notation $\vec{x}_j$ ($\vec{x}'_j$, resp.) refers to the variable block $\vec{x}_j(1) \dots \vec{x}_j(\kappa)$ ($\vec{x}'_j(1) \dots \vec{x}'_j(\kappa')$, resp.). By premise, we know that $|\vec{x}_j(i)|\ge 2$ for each $j\in\{2,\ldots,n\}$ and each $i\in\{1,\ldots,\kappa\}$.
    
    Let $y_1$, \dots, $y_\kappa$ be variables in $p$ such that $[y_1]$, \dots, $[y_\kappa]$ are the distinct variable groups in $(p,R)$.

We first define a property of substitutions that will be crucial for the purpose of this proof. For $i\in\{1,\ldots,\kappa\}$, an $R$-substitution $\theta_i$ is called \emph{$(p,i)$-unambiguous}\/ if the following conditions are satisfied.
\begin{enumerate}
\item $\theta_i(x)=\varepsilon$ if $x\notin[y_i]$, and $|\theta_i(x)|=1$ if $x\in [y_i]$. In particular, $\theta_i$ is an $\ell_1$-substitution with 
\[
\theta_i(p)=\ s_{i,1}\,\omega_1 s_{i,2}\,\omega_2\ldots s_{i,n}\,\omega_n s_{i,n+1}\,,
\]
for some words $s_{i,j}=\theta_i(\vec{x}_j)=\theta_i(\vec{x}_j(i))$, where $|s_{i,j}|=|\vec{x}_j(i)|$.
\item For $j\in \{1,\ldots,n\}$, the words $s_{i,j}$ obey the following constraints:
\begin{enumerate}
\item $\omega_j$ is not a substring of $s_{i,j}\circ \omega_j[1:|\omega_j|-1]$.
\item $\omega_j$ is not a substring of $\omega_j[2:|\omega_j|]\circ s_{i,j+1}$.
\end{enumerate}
\end{enumerate}

The proof is organized along the following three claims.

\noindent\emph{Claim 1.} If $\theta_i$ is $(p,i)$-unambiguous, then there is an $R'$-substitution $\theta'_i$ such that $\theta_i(p)=\theta'_i(p')$. Moreover, \emph{every}\/ such substitution $\theta'_i$ satisfies
 \begin{equation}\label{eqn:theta'theta}
\theta'_i(\vec{x}'_j)=\theta_i(\vec{x}_j)=\theta_i(\vec{x}_j(i))\mbox{ for all }j\in\{1,\ldots,n+1\}\,.
   \end{equation}

   \medskip
   \noindent\emph{Claim 2.} Let $\theta_i$ be $(p,i)$-unambiguous, and let $\theta'_i$ be an $R'$-substitution with $\theta_i(p)=\theta'_i(p')$. Then the validity of Equation~(\ref{eqn:theta'theta}) for all $i\in\{1,\ldots,\kappa\}$ implies that $\varepsilon L_{len}(p,R) \subseteq \varepsilon L_{len}(p',R')$.

   \medskip

   \noindent\emph{Claim 3.} For each $i\in\{1,\ldots,\kappa\}$, a $(p,i)$-unambiguous $R$-substitution exists.

   \medskip

   Clearly, proving these three claims completes the proof of Theorem~\ref{thm:congruous}. In the proof of these claims, we use $\sqsubseteq$ to denote the substring relation, and $\sqsubset$ to denote the proper substring relation.

   \medskip

   \noindent\emph{Proof of Claim 1.} 
Since $\mathcal{S}_{\varepsilon,1}(p,R) \subseteq \varepsilon L_{len}(p',R')$, we obtain $\theta_i(p)\in \varepsilon L_{len}(p',R')$. In particular, there exists an $R'$-substitution $\theta'_i$ with  $\theta'_i(p')=\theta_i(p)$. Now let $\theta'_i$ be any such substitution. Fix $i\in\{1,\ldots,\kappa\}$. By definition, $\theta_i(\vec{x}_j)=\theta_i(\vec{x}_j(i))$. To verify that Equation~(\ref{eqn:theta'theta}) holds, suppose the opposite. Since $\theta'_i(p')=\theta_i(p)$, this is equivalent to saying that there is some $j\in\{1,\ldots,n\}$ with $\theta'_i(\vec{x}'_1\omega_1\ldots\omega_{j-1}\vec{x}'_j)\ne\theta_i(\vec{x}_1\omega_1\ldots\omega_{j-1}\vec{x}_j)$.

Suppose 
\begin{equation}\label{eqn:induction}
\theta'_i(\vec{x}'_1\omega_1\ldots\omega_{j-1}\vec{x}'_j)\sqsupset\theta_i(\vec{x}_1\omega_1\ldots\omega_{j-1}\vec{x}_j)\,.
\end{equation}
We will show that Inequality~(\ref{eqn:induction}) implies 
\begin{equation}\label{eqn:induction2}
\theta'_i(\vec{x}'_1\omega_1\ldots\omega_{j}\vec{x}'_{j+1})\sqsupset\theta_i(\vec{x}_1\omega_1\ldots\omega_{j}\vec{x}_{j+1})\,.
\end{equation}
Consider two cases:

\medskip

\noindent\emph{Case 1.} $\theta'_i(\vec{x}'_1\omega_1\ldots\omega_{j-1}\vec{x}'_j)\sqsupset\theta_i(\vec{x}_1\omega_1\ldots\omega_{j-1}\vec{x}_j\omega_j)$.

Since $\theta_i$ is $(p,i)$-unambiguous, we know that $\omega_j\not\sqsubseteq s_{i,j+1}=\theta_i(\vec{x}_{j+1})$. This implies 
\[\theta'_i(\vec{x}'_1\omega_1\ldots\omega_{j-1}\vec{x}'_j\omega_j)\sqsupset\theta_i(\vec{x}_1\omega_1\ldots\omega_{j-1}\vec{x}_j\omega_j\vec{x}_{j+1})\,,\]
and hence 
\[\theta'_i(\vec{x}'_1\omega_1\ldots\omega_{j-1}\vec{x}'_j\omega_j\vec{x}'_{j+1})\sqsupset\theta_i(\vec{x}_1\omega_1\ldots\omega_{j-1}\vec{x}_j\omega_j\vec{x}_{j+1})\,,\]
as desired, i.e., we obtain Inequality~(\ref{eqn:induction2}).

\medskip

\noindent\emph{Case 2.} Not Case 1. Thus we get 
\begin{eqnarray*}
    \theta'_i(\vec{x}'_1\omega_1\ldots\omega_{j-1}\vec{x}'_j)&\sqsubseteq&\theta_i(\vec{x}_1\omega_1\ldots\omega_{j-1}\vec{x}_j\omega_j)\mbox{ and}\\
    \theta'_i(\vec{x}'_1\omega_1\ldots\omega_{j-1}\vec{x}'_j)&\sqsupset&\theta_i(\vec{x}_1\omega_1\ldots\omega_{j-1}\vec{x}_j)\,.
\end{eqnarray*}
The latter implies
\[
\theta'_i(\vec{x}'_1\omega_1\ldots\omega_{j-1}\vec{x}'_j\omega_j)\sqsupset\theta_i(\vec{x}_1\omega_1\ldots\omega_{j-1}\vec{x}_j\omega_j)\,.
\]

If Inequality~(\ref{eqn:induction2}) is not satisfied, we get 
\[\theta'_i(\vec{x}'_1\omega_1\ldots\omega_{j-1}\vec{x}'_j\omega_j\vec{x}'_{j+1})\sqsubseteq\theta_i(\vec{x}_1\omega_1\ldots\omega_{j-1}\vec{x}_j\omega_j\vec{x}_{j+1})\,,\]
which implies $\theta'_i(\vec{x}'_{j+1})\sqsubset\theta_i(\vec{x}_{j+1})=s_{i,j+1}$ and thus $\omega_j\sqsubseteq\omega_j[2:|\omega_j|]\circ s_{i,j+1}$. This is impossible since $\theta_i$ is $(p,i)$-unambiguous. Hence Inequality~(\ref{eqn:induction2}) is satisfied. 

\medskip

Inductively, we thus obtain 
\begin{equation*}
\theta'_i(\vec{x}'_1\omega_1\ldots\omega_{k-1}\vec{x}'_{k})\sqsupset\theta_i(\vec{x}_1\omega_1\ldots\omega_{k-1}\vec{x}_{k})\mbox{ for all }k\ge j\,,
\end{equation*}
which contradicts $\theta'_i(p')=\theta_i(p)$.   

By a completely symmetric argument, if we reverse the substring relation in Inequality~(\ref{eqn:induction}) and assume 
\begin{equation*}
\theta'_i(\vec{x}'_1\omega_1\ldots\omega_{j-1}\vec{x}'_j)\sqsubset\theta_i(\vec{x}_1\omega_1\ldots\omega_{j-1}\vec{x}_j)\,,
\end{equation*}
we will obtain 
\begin{equation*}
\theta'_i(\vec{x}'_1\omega_1\ldots\omega_{k-1}\vec{x}'_{k})\sqsubset\theta_i(\vec{x}_1\omega_1\ldots\omega_{k-1}\vec{x}_{k})\mbox{ for all }k\ge j\,,
\end{equation*}
which again contradicts $\theta'_i(p')=\theta_i(p)$.

Thus, for all $j$,
\begin{equation*}
\theta'_i(\vec{x}'_1\omega_1\ldots\omega_{j-1}\vec{x}'_j)=\theta_i(\vec{x}_1\omega_1\ldots\omega_{j-1}\vec{x}_j)\,,
\end{equation*}
which finally proves Claim 1.
\hfill $\Box$ \emph{(Claim 1.)}

   \medskip

   \noindent\emph{Proof of Claim 2.} 
 Equation~(\ref{eqn:theta'theta}) implies that there exists a vector $(t^i_1,\ldots,t^i_{\kappa'})\in\mathbb{N}^{\kappa'}$ such that  
   \begin{equation}\label{eqn:lincomb}
    |\vec{x}_j(i)| = \sum_{\eta =1}^{\kappa'} t^i_\eta |\vec{x}'_j(\eta)|\mbox{ for all }j\in\{1,\ldots,n+1\}\,,    
   \end{equation}
    where $(t^i_1,\ldots,t^i_{\kappa'})\in\mathbb{N}^{\kappa'}$ is independent of $j$. This is true because $|\theta_i(x)|=1$ iff $x\in [y_i]$.

    To obtain $\varepsilon L_{len}(p,R)\subseteq \varepsilon L_{len}(p',R')$, consider any word $s\in \varepsilon L_{len}(p,R)$, generated by an $R$-substitution $\theta$ with $\theta(p)=s$. Our goal is to define an $R'$-substitution $\theta'$ that generates $s$ from $p'$. We will define $\theta'$ such that 
    \[
\theta'(\vec{x}'_j)=\theta'(\vec{x}'_j(1)\ldots\vec{x}'_j(\kappa'))=\theta(\vec{x}_j(1)\ldots\vec{x}_j(\kappa))=\theta(\vec{x}_j)\mbox{ for all }j\,,
\]
    which immediately yields $\theta'(p')=\theta(p)=s$.     We only need to verify that $\theta'$ can be defined this way without violating the relations in $R'$. The latter can be seen as follows.
    
    First note that 
    \[
|\theta(\vec{x}_j)|=\sum_{i=1}^{\kappa}|\theta(\vec{x}_j(i))|=\sum_{i=1}^{\kappa}r^i|\vec{x}_j(i)|=\sum_{i=1}^{\kappa} r^i \sum_{\eta =1}^{\kappa'}  t^i_\eta |\vec{x}'_j(\eta)|=\sum_{\eta =1}^{\kappa'} \sum_{i=1}^{\kappa} r^i  t^i_\eta |\vec{x}'_j(\eta)|\,,
    \]
    where $r^i=|\theta(y_i)|$.

    Let $y'_1$, \dots, $y'_{\kappa'}$ be variables in $p'$ such that $[y'_1]$, \dots, $[y'_{\kappa'}]$ are the distinct variable groups in $(p',R')$. Fix $i\le\kappa$. For any $\eta\le\kappa'$ and all $y\in [y'_\eta]$, we define the length of the desired word $\theta'(y)$ to be equal to $\sum_{i=1}^{\kappa} r^i  t^i_\eta$. Thus
    \[
|\theta'(\vec{x}'_j)|=\sum_{\eta=1}^{\kappa'}|\theta'(\vec{x}'_j(\eta))|=\sum_{\eta =1}^{\kappa'} \sum_{i=1}^{\kappa} r^i  t^i_\eta |\vec{x}'_j(\eta)|\,.
    \]
    Any substitution $\theta'$ satisfying all these length constraints is valid for $R'$. Moreover, since $|\theta'(\vec{x}'_j)|=|\theta(\vec{x}_j)|$ for all $j$, one can obviously define $\theta'$ in a way that $\theta'(\vec{x}'_j)=\theta(\vec{x}_j)$ for all $j$, without violating the relations in~$R'$.
   \hfill $\Box$ \emph{(Claim 2.)}

   \medskip

   \noindent\emph{Proof of Claim 3.} Fix $i\in\{1,\ldots,\kappa\}$ and consider an $\ell_1$-substitution $\theta_i$ valid for $(p,R)$ defined as follows. 
\begin{align*}
    \theta_i (x)& =\varepsilon\mbox{ if } x \not\in [x_i]\,,\\
    \theta_i (\vec{x}_j(i)) &=s_{i,j}\,,
\end{align*}
where $s_{i,1}=\sigma^{|\vec{x}_{1}(i)|}$ for $\sigma\ne\omega_1[1]$, $s_{i,n+1}=\sigma^{|\vec{x}_{n+1}(i)|}$ for $\sigma\ne\omega_n[|\omega_n|]$, and, for $j\in\{2,\ldots,n\}$, the word $s_{i,j}$ is a string of length $|\vec{x}_j(i)|$ defined by
    \[
        s_{i,j} =\begin{cases}
                    
        (\sigma \sigma')^{\frac{|\vec{x}_{j}(i)|}{2}}, &  \omega_{j-1} \in \{\sigma'\}^+,\ \omega_{j} \in \{\sigma\}^+,\text{ and } |\Vec{x}_{j}(i)| \mbox{ is even}\,,\\

        (\sigma \sigma')^{\frac{|\vec{x}_{j}(i)|-1}{2}} \sigma', &  \omega_{j-1} \in \{\sigma'\}^+,\  \omega_{j} \in \{\sigma\}^+, \text{ and } |\Vec{x}_{j}(i)|\mbox{ is odd}\,,\\

       \sigma^{|\vec{x}_{j}(i)|}, & \omega_{j-1},\omega_{j} \in \{\sigma'\}^+\,,\\

        \sigma^{|\vec{x}_{j}(i)|}, & \omega_{j-1} \in \{\sigma'\}^+, \mbox{ and } \omega_{j}\mbox{ contains both }a\mbox{ and }b\,,\\

        \sigma^{|\vec{x}_{j}(i)|}, & \omega_{j} \in \{\sigma'\}^+, \mbox{ and }  \omega_{j-1}\mbox{ contains both }a\mbox{ and }b\,,\\
        
        \sigma^{|\vec{x}_{j}(i)|}, & \omega_{j-1},\omega_{j}\mbox{ both contain both }a\mbox{ and }b, \text{ and } \sigma \in\Sigma\mbox{ arbitrary}\,,
		\end{cases}
    \]
        where in each of the first five cases of the case distinction, $\sigma$ and $\sigma'$ are chosen such that $\Sigma=\{a,b\}=\{\sigma,\sigma'\}$; in particular, in every such case, $\sigma \neq \sigma'$. 

Clearly, Condition 1 of $(p,i)$-unambiguity is fulfilled. We only need to verify Condition 2. Since $s_{i,1}$  only has the letter different from the first letter in $\omega_1$, we know that $\omega_1$ is not a substring of $s_{i,1}\circ \omega_1[1:|\omega_1|-1]$. Likewise, $s_{i,n+1}$ only having the letter different from the last letter in $\omega_n$ implies that $\omega_n$ is not a substring of $\omega_n[2:|\omega_n|]\circ s_{i,n+1}$. Fix $j\in\{2,\ldots,n\}$ and consider two cases for $s_{i,j}$. 

First, suppose $s_{i,j}\notin\{\sigma\}^+$ for any $\sigma\in\Sigma$. This means there are $\sigma,\sigma'\in\Sigma$, $\sigma\ne\sigma'$ such that  $\omega_{j-1}\in\{\sigma'\}^+$ and $\omega_j\in\{\sigma\}^+$. Since $\omega_{j-1},\omega_j$ are of length at least 3, the shape of the word $s_{i,j}$ ensures that  $\omega_j$ is not a substring of $s_{i,j}\circ \omega_j[1:|\omega_j|-1]$, and $\omega_{j-1}$ is not a substring of $\omega_{j-1}[2:|\omega_{j-1}|]\circ s_{i,j}$. 

Second, suppose $s_{i,j}\in\{\sigma\}^+$ for some $\sigma\in\Sigma$. Consider $\omega_j$. By definition of $s_{i,j}$, either $\omega_j\in\{\sigma'\}^+$ for $\sigma'\ne\sigma$, or $\omega_j$ contains both $a$ and $b$. For either option, it is obvious that $\omega_j \not\sqsubseteq s_{i,j}$. Now suppose $\omega_j$ is a substring of $s_{i,j}\circ \omega_j[1:|\omega_j|-1]$. Then there is some $t$, $1 \leq t \leq \min \{ |\omega_j|-1, |s_{i,j}| \}$, such that $\omega_j=s_{i,j}[|s_{i,j}|-t+1: |s_{i,j}|] \circ \omega_j[1:|\omega_j|-t]=\sigma^t\circ \omega_j[1:|\omega_j|-t]$. Hence $\omega_j\notin\{\sigma'\}^+$, where $\sigma'\ne\sigma$, so that $\omega_j$ contains both $\sigma$ and $\sigma'$. Let $\tau\ge 0$ so that $\omega_j$ starts with $\sigma^{\tau}\sigma'$. 
Now the first position in $\omega_j$ at which $\sigma'$ occurs is $\tau+1$. So the first position in $\sigma^t \circ \omega_j[1:|\omega_j|-t]$ at which $\sigma'$ occurs is $t+\tau+1$. But then, since $\omega_j=\sigma^t\circ \omega_j[1:|\omega_j|-t]$, the first position in $\omega_j$ at which $\sigma'$ occurs is $t+\tau+1$. This implies $\tau+1=t+\tau+1$, contradicting $t>0$.
Therefore, $\omega_j$ is not a substring of $s_{i,j}\circ \omega_j[1:|\omega_j|-1]$. By symmetric reasoning, $\omega_{j-1}$ is not a substring of $\omega_{j-1}[2:|\omega_{j-1}|]\circ s_{i,j}$.
   \hfill $\Box$ \emph{(Claim 3.)}

\end{proof}

\section{Conclusions}
This paper provided insights into the connections between learning with characteristic sets and learning with  telltales, via the notion of positive characteristic set. Its main contribution is to provide new (non-)learnability results for erasing relational pattern languages, with a focus on positive characteristic sets and telltales.

To show that erasing pattern languages under the reversal relation have no positive characteristic sets (for binary alphabets), we used a construction that Reidenbach~\cite{Reidenbach02} devised to show the corresponding statement for the equality relation. An interesting open question from this construction is the following: Is it true that a relational pattern $(p,R)$ has a telltale with respect to $\varepsilon\mathcal{L}_{rev}$ iff it has a telltale with respect to $\varepsilon\mathcal{L}_{eq}$? If yes, is there a constructive proof that shows how to transform telltales for either relation into telltales for the other?

Our main result implies that, for binary alphabets, linear-size telltales exist (and thus positive characteristic sets exist) for erasing pattern languages under the equal length relation, at least when limiting the patterns to those whose terminal blocks have length at least 3 and are not of the shape $a^nb$ or $ab^n$, and whose relations satisfy an additional structural constraint. We leave open whether or not these (admittedly rather limiting) conditions can be removed. A related open question is whether positive characteristic sets also exist in the non-erasing case when the alphabet is binary. Our proof makes use of the erasing property, and might need substantial adjustments for the non-erasing case.

Different real-world applications might give rise to different relations to study in the context of learning relational pattern languages. This paper analyzed only two of a multitude of possibilities; much future work is needed to gain a better understanding of the effects of various kinds of relations on learnability.

\bmhead{Acknowledgements} We are deeply grateful to an anonymous reviewer for their exceptionally thorough review and their very helpful comments, which helped correct multiple technical problems in an earlier version of the paper; in particular, we thank the reviewer for their suggestion of presenting the proof of Lemma~\ref{lem:ex-|alphabet|=2-length-5} in its current form. We also thank R.\ Holte for valuable feedback. 

S.\ Zilles was supported through a Canada CIFAR AI Chair at the Alberta Machine Intelligence Institute (Amii), through a Natural Sciences and Engineering Research Council (NSERC) Canada Research Chair, through the New Frontiers in Research Fund (NFRF) under grant  no.\ NFRFE-2023-00109 and through the NSERC Discovery Grants program under grant no.\ RGPIN-2017-05336.

%
%
%
%
 \bibliography{bib}

\begin{appendices}

\section{Proof of Lemma \ref{lem:ex-|alphabet|=2-length-5}}

\setcounter{lemma}{9}
\begin{lemma}
  Let $(p,R)$ and $(p',R')$ be arbitrary patterns over $\Sigma = \{a,b\}$, where $p$ and $p'$ are incongruous. Suppose there do not exist telltale conjugates $\pi$, $\pi'$ such that 
    \begin{itemize}
        \item $\pi$ is a substring of $p$ that is not preceded or succeeded by terminal symbols in $p$, and
        \item $\pi'$ is a substring of $ p'$  that is not preceded or succeeded by terminal symbols in $p'$.
    \end{itemize}  
    Let  $\mathcal{S}_{\varepsilon,2}(p,R) \subseteq \varepsilon L_{len}(p',R') \subseteq \varepsilon L_{len}(p,R)$. Then $\varepsilon L_{len}(p,R)=\varepsilon L_{len}(p',R')$. 
\end{lemma}

\medskip

In the proof of Lemma~\ref{lem:ex-|alphabet|=2-length-5}, every $\ell_z$-substitution replaces only the variables in a single group with a word of length $z$; all other variables are erased (replaced by $\varepsilon$). Hence, the $\ell_z$-substitutions that we consider generate words in $\mathcal{S}_{\varepsilon,z}(p,R)$.

Note that we use $\sqsubseteq$ to denote the substring relation, and $\sqsubset$ to denote the proper substring relation.

\begin{proof}
    Let $n, n'\in\mathbb{N}$ such that 
    \[p = \vec{x}_1 \omega_1 \vec{x}_2 \omega_2 \ldots \vec{x}_n \omega_n \vec{x}_{n+1} \]
    and 
    \[ p' = \vec{x}'_1 \omega'_1 \vec{x}'_2 \omega'_2 \ldots \vec{x}'_{n'} \omega'_{n'} \vec{x}'_{n'+1}\,, \]
    where $\vec{x}_1, \vec{x}'_1, \vec{x}'_{n'+1}, \vec{x}_{n+1} \in X^\ast$ and $\vec{x}_2, , \ldots, \vec{x}_n, \vec{x}'_2, \ldots, \vec{x}'_n \in X^+$, as well as $\omega_i,\omega'_j\in\Sigma^+$ for all $i,j$.
    
    Since $\mathcal{S}_{\varepsilon,2}(p,R) \subseteq \varepsilon L_{len}(p',R') \subseteq \varepsilon L_{len}(p,R)$, we have $\omega_1 \ldots \omega_n = \omega'_1 \ldots \omega'_{n'}$. If $n=n'$ and $\omega_i=\omega'_i$ for all $i\le n$, then $\mathcal{S}_{\varepsilon,2}(p,R) \subseteq \varepsilon L_{len}(p',R') \subseteq \varepsilon L_{len}(p,R)$ implies that, for $t\in\{1,n+1\}$, either both or none of $\vec{x}_t,\vec{x}'_t$ are empty. This contradicts the premise that $p$ and $p'$ are incongruous. Hence assume, by way of contradiction, that $i<\min\{n,n'\}$ is smallest  such that $\omega'_i \neq \omega_i$. We have either $\omega_i \sqsubset \omega'_i$ or $\omega'_i \sqsubset \omega_i$. Suppose that $\omega'_i \sqsubset \omega_i$. (The proof is completely analogous if $\omega_i \sqsubset \omega'_i$.)
    This means that $|\omega_i| > |\omega'_i|$ and $\omega'_i$ is a proper prefix of $\omega_i$. 
    Our goal is now to derive a contradiction; in particular we will deduce that 
    \[\mathcal{S}_{\varepsilon, 2}(p,R) \not\subseteq \varepsilon L_{len}(p',R')\mbox{  or } \varepsilon L_{len}(p',R') \not\subseteq \varepsilon L_{len}(p,R)\,, \hspace*{2cm} (\star)\] 
    which contradicts the premises of the lemma.
    
    Now we assume, without loss of generality, that $\omega_i$ ends in $b$. Consider the following case distinction.

   \emph{Case 1.} $\omega'_i$ contains both $a$ and $b$.

   Note that $\omega'_i$ is a proper prefix of $\omega_i$. Let $\{\sigma,\overline{\sigma}\}=\{a,b\}$ and let $\omega'_i\sigma$ be a prefix of $\omega_i$. Then consider the $R'$-substitution $\theta'$ that replaces one variable $x$ in $\vec{x}'_{i+1}$ (as well as all variables in $[x]$) with $\overline{\sigma}$ and erases all other variables. We obtain $\theta'(p')\in\varepsilon L_{len}(p',R')\setminus \varepsilon L_{len}(p,R)$, yielding $(\star)$.  

   To see that $\theta'(p')\notin \varepsilon L_{len}(p,R)$, note that any $R$-substitution generating $\theta'(p')$ from $p$ must replace each variable in $p$ with a word from $\{ \overline{\sigma} \}^\ast$. Let $k$ be such that the rightmost occurrence of $\sigma$ in $\omega'_i$ corresponds to the $k$th occurrence of the letter $\sigma$ in $\theta'(p')$, and thus the first position of $\omega'_{i+1}$ corresponds to the $(k+1)$st occurrence of the letter $\sigma$ in $\theta'(p')$.
   Let $e$ be the number of occurrences of the letter $\overline{\sigma}$ between the $k$th and $(k+1)$st occurrence of $\sigma$ in $\theta'(p')$. Note that $e \ge 1$.
   Since any $R$-substitution $\theta$ generating $\theta'(p')$ from $p$ must replace each variable in $p$ with a word from $\{ \overline{\sigma} \}^\ast$, 
   the $k$th and $(k+1)$st occurrences of the letter $\sigma$ in $\theta(p)$ are part of $\omega_i$. Note that $\omega_i \in \omega'_i \sigma \Sigma^\ast$. However, this implies that $\theta(p)$ has fewer than $e$ occurrences of the letter $\overline{\sigma}$ between the $k$th and $(k+1)$st occurrences of the letter $\sigma$. Therefore, no $R$-substitution can generate $\theta'(p')$ from $p$. Thus $\theta'(p')\in\varepsilon L_{len}(p',R')\setminus \varepsilon L_{len}(p,R)$ and $(\star)$ holds.

   \emph{Case 2.} $\omega'_i=b^n$ for some $n$, and $\omega_i$ starts with $b^{n+1}$.

    Now consider an $R'$-substitution $\theta'$ that replaces one variable $x$ in $\vec{x}'_{i+1}$ (as well as all variables in $[x]$) with $a$ and erases all other variables. 

    Let $k$ be such that the rightmost occurrence of $b$ in $\omega'_i$ corresponds to the $k$th occurrence of the letter $b$ in $\theta'(p')$.
    Since any $R$-substitution $\theta$ generating $\theta'(p')$ from $p$ must replace each variable in $p$ with a word from $\{ a \}^\ast$, and $\omega_i$ starts with $b^{n+1}$, the letter $b$ occurs immediately to the right of the $k$th occurrence of the letter $b$ in $\theta(p)$, while the letter $a$ occurs immediately to the right of the $k$th occurrence of the letter $b$ in $\theta'(p')$. Therefore, no $R$-substitution can generate $\theta'(p')$ from $p$. Thus $\theta'(p')\in\varepsilon L_{len}(p',R')\setminus \varepsilon L_{len}(p,R)$ and $(\star)$ holds.


     \emph{Case 3.} $\omega'_i=b^n$ for some $n$, and $\omega_i$ is of the form $b^na^m b\Sigma^*$ for some $m>0$.
    
    Consider an $R'$-substitution $\theta'$ that replaces one variable $x$ in $\vec{x}'_{i+1}$ (as well as all variables in $[x]$) with $a$ and erases all other variables. 
    
    Let $k$ be such that the last occurrence of $b$ in $\omega'_i$ corresponds to the $k$th occurrence of the letter $b$ in $\theta'(p')$.
    Let $m$ be the number of occurrences of the letter $a$ in $\omega_i$ that are to the left of the $(n+1)$st occurrence of $b$ in $\omega_i$.  Obviously, there are more than $m$ occurrences of $a$ between the $k$th and $(k+1)$st occurrences of $b$ in $\theta'(p')$. On the other hand, since any $R$-substitution $\theta$ generating $\theta'(p')$ from $p$ must replace each variable in $p$ with a word from $\{ a \}^\ast$, the number of occurrences of $a$ between the $k$th and $(k+1)$st occurrences of $b$ in $\theta(p)$ remains $m$. Thus $\theta'(p')\in\varepsilon L_{len}(p',R')\setminus \varepsilon L_{len}(p,R)$ and $(\star)$ holds.

    \emph{Case 4.} $\omega'_i=a^n$ for some $n$, and $\omega_i$ is of the form $a^{n+1} \Sigma^\ast b$.

    Analogous to Case 2, where $a$ and $b$ are interchanged.

    \emph{Case 5.} $\omega'_i=a^n$ for some $n$, and $\omega_i$ is of the form $a^n b^m a \Sigma^\ast$ for some $m > 0$. 

    Analogous to Case 3, where $a$ and $b$ are interchanged.

    \emph{Case 6.} $\omega'_i=a^n$ and $\omega_i$ is of the form $a^n b^m$ for some $n, m \geq 2$. 

     Consider an $R'$-substitution $\theta'$ that replaces the first variable $x$ of $\vec{x}'_{i+1}$ by $ba$, every other variable in $[x]$ by $bb$, while erasing all remaining variables.

Let $k$ be such that $\omega'_1\cdots\omega'_{i}$ contains the letter $a$ exactly $k$ times. Then, the word $\theta'(p')$ contains the substring $ba$ immediately after its $k$th occurrence of the letter $a$. Now suppose $\theta$ is an $R$-substitution with $\theta(p)=\theta'(p')$. Clearly, $\theta$ replaces exactly one variable with a word containing the letter $a$, and this word contains exactly one $a$. Moreover, $\omega_1\cdots\omega_{i}$ contains the letter $a$ exactly $k$ times. Therefore, the word
$$
u:=\theta(\vec{x}_1\omega_1\vec{x}_2\omega_2\cdots\omega_{i-1}\vec{x}_i\omega_i)
$$
has either exactly $k$ or exactly $k+1$ occurrences of the letter $a$.

If $u$ has exactly $k+1$ occurrences of the letter $a$, then the fact that $\omega_i=a^nb^m$ with $n,m\ge 2$ implies that the $(k+1)$st occurrence of $a$ in $\theta(p)$ is immediately preceded by another $a$. By contrast, clearly, the $(k+1)$st occurrence of $a$ in $\theta'(p')$ is immediately preceded by the letter $b$, which contradicts $\theta(p)=\theta'(p')$. If $u$ has exactly $k$ occurrences of the letter $a$, then the $k$th occurrence of $a$ in $\theta(p)$ is immediately followed by $bb$, since $m\ge 2$. By contrast, the $k$th occurrence of $a$ in $\theta'(p')$ is immediately followed by $ba$, which again contradicts $\theta(p)=\theta'(p')$.
    Thus $\theta'(p')\in\varepsilon L_{len}(p',R')\setminus \varepsilon L_{len}(p,R)$ and $(\star)$ holds.

    \emph{Case 7.} $\omega'_i=a^{n_i}$ and $\omega_i$ is of the form $a^{n_i} b^{m_i}$ for some $n_i = 1$ or $m_i = 1$. 
    
    Let $j\ge i$ be maximal such that, for all $k\in\{i,\ldots,j\}$, we have 
    $\omega'_k = b^{m_{k-1}} a^{n_k}$ and $\omega_k = a^{n_k} b^{m_k}$ where $n_k, m_k$ are as in Definition~\ref{def:conjugate-pat}, except that $m_{i-1}:=0$.

Now consider the following sub-cases:
    \begin{itemize}

        \item[7.1.] $\omega'_{j+1} = b^{m_j+1} \Sigma^\ast $.

        Consider an $R$-substitution $\theta$ that replaces one variable $x$ in $\vec{x}_{j+1}$ (as well as all other variables in $[x]$) with $a$ and erases all other variables. 

        Let $k$ be such that the last occurrence of $b$ in $\omega_j$ corresponds to the $k$th occurrence of the letter $b$ in $\theta(p)$.
        Suppose there is an $R'$-substitution $\theta'$ that generates $\theta(p)$ from $p'$. Since such $\theta'$ must replace each variable in $p'$ with a word from $\{ a \}^\ast$ and $\omega'_{j+1}$ starts with $b^{m_j+1}$, the letter $b$ occurs immediately to the right of the $k$th occurrence of the letter $b$ in $\theta'(p')$, while the letter $a$ occurs immediately to the right of the $k$th occurrence of the letter $b$ in $\theta(p)$. This is a contradiction to $\theta(p)=\theta'(p')$.  Hence, we obtain $\theta(p)\in\mathcal{S}_{\varepsilon, 2}(p,R)\setminus \varepsilon L_{len}(p', R')$, yielding $(\star)$. 

        \item[7.2.] $\omega'_{j+1} = b^m$ for $m < m_j$.

        Consider an $R'$-substitution $\theta'$ that replaces one variable $x$ in $\vec{x}'_{j+2}$ (as well as all other variables in $[x]$) with $a$ and erases all other variables.
        
        Let $k$ be such that the last occurrence of $b$ in $\omega'_{j+1}$ corresponds to the $k$th occurrence of the letter $b$ in $\theta'(p')$.
        Suppose there is an $R$-substitution $\theta$ that generates $\theta'(p')$ from $p$. Since such $\theta$ must replace each variable in $p$ with a word from $\{ a \}^\ast$ and $\omega_{j}$ ends with $b^{m_j}$ for $m_j > m$, the letter $b$ occurs immediately to the right of the $k$th occurrence of the letter $b$ in $\theta(p)$, while the letter $a$ occurs immediately to the right of the $k$th occurrence of the letter $b$ in $\theta'(p')$. This is a contradiction to $\theta(p)=\theta'(p')$.  Hence, we obtain $\theta(p)\in \varepsilon L_{len}(p', R') \setminus \varepsilon L_{len} (p, R)$, yielding $(\star)$.

        \item[7.3.] $\omega'_{j+1} = b^{m_j}$. 
        
        Then $\omega_i \vec{x}_{i+1} \omega_{i+1} \ldots \vec{x}_j \omega_j$ and $\omega'_i \vec{x}'_{i+1} \omega'_{i+1} \ldots \vec{x}'_{j+1} \omega'_{j+1}$ are telltale conjugate subpatterns of $p$, resp.\ $p'$. This contradicts the premise of the Lemma.

        \item[7.4.] $\omega'_{j+1} = b^{m_j} a^{n} b \Sigma^\ast$ for some $n >0$.

        Consider an $R$-substitution $\theta$ that replaces one variable $x$ in $\vec{x}_{j+1}$ (as well as all variables in $[x]$) with $a$ and erases all other variables.
        
        In particular, let $k$ be such that the last occurrence of $b$ in $\omega_j$ corresponds to the $k$th occurrence of $b$ in $\theta(p)$. 
        Obviously, there are more than $n$ occurrences of $a$ between the $k$th and $(k+1)$st occurrences of $b$ in $\theta(p)$. On the other hand, since any $R$-substitution $\theta'$ generating $\theta(p)$ from $p'$ must replace each variable in $p'$ with a word from $\{ a \}^\ast$, the number of occurrences of $a$ between the $k$th and $(k+1)$st occurrences of $b$ in $\theta'(p')$ remains $n$. Thus $\theta(p)\in\mathcal{S}_{\varepsilon, 2}(p,R)\setminus \varepsilon L_{len}(p',R')$ and $(\star)$ holds.

        \item[7.5.] $\omega'_{j+1} = b^{m_j} a^{n}$ for some $n > 0$ and $\omega_{j+1}$ starts with $a^t$ for some $t > n$.

        Consider an $R'$-substitution $\theta'$ that replaces one variable $x$ in $\vec{x}'_{j+2}$ (as well as all variables in $[x]$) with $b$ and erases all other variables. 

        Let $k$ be such that the last occurrence of $a$ in $\omega'_{j+1}$ corresponds to the $k$th occurrence of the letter $a$ in $\theta'(p')$.
        Suppose there is an $R$-substitution $\theta$ that generates $\theta'(p')$ from $p$. Since such $\theta$ must replace each variable in $p$ with a word from $\{ b \}^\ast$ and $\omega_{j+1}$ starts with $a^t$ for $t > n$, the letter $a$ occurs immediately to the right of the $k$th occurrence of the letter $a$ in $\theta(p)$, while the letter $b$ occurs immediately to the right of the $k$th occurrence of the letter $a$ in $\theta'(p')$. This is a contradiction to $\theta(p)=\theta'(p')$.  Hence, we obtain $\theta'(p')\in \varepsilon L_{len}(p',R') \setminus \varepsilon L_{len}(p, R)$, yielding $(\star)$.

        \item[7.6.] $\omega'_{j+1} = b^{m_j} a^{n}$ for some $n$ and $\omega_{j+1} = a^t$ for some $t < n$.

        Consider an $R$-substitution $\theta$ that replaces one variable $x$ in $\vec{x}_{j+2}$ (as well as all variables in $[x]$) with $b$ and erases all other variables.

        Let $k$ be such that the last occurrence of $a$ in $\omega_{j+1}$ corresponds to the $k$th occurrence of the letter $a$ in $\theta(p)$.
        Suppose there is an $R'$-substitution $\theta'$ that generates $\theta(p)$ from $p'$. Since such $\theta'$ must replace each variable in $p'$ with a word from $\{ b \}^\ast$ and $\omega'_{j+1} = b^{m_j} a^n$ for $n >  t$, the letter $a$ occurs immediately to the right of the $k$th occurrence of the letter $a$ in $\theta'(p')$, while the letter $b$ occurs immediately to the right of the $k$th occurrence of the letter $a$ in $\theta(p)$. This is a contradiction to $\theta(p)=\theta'(p')$.  Hence, we obtain $\theta(p)\in\mathcal{S}_{\varepsilon, 2}(p,R)\setminus \varepsilon L_{len}(p', R')$, yielding $(\star)$. 
        
        \item[7.7.] $\omega'_{j+1} = b^{m_j} a^{n}$ and $\omega_{j+1} = a^n$ where $n, m_j \geq 2$.

        We use an argument similar to that used in Case 6. Consider an $R$-substitution $\theta$ that replaces the first variable $x$ in $\vec{x}_{j+1}$ with $ab$, all other variables in $[x]$ with $aa$, and erases all other variables.

Let $k$ be such that $\omega_1 \cdots \omega_{j}$ contains the letter $b$ exactly $k$ times. 
        Then, the word $\theta(p)$ contains the substring $ab$ immediately after its $k$th occurrence of the letter $b$.

      Suppose $\theta'$ is an $R'$-substitution that generates $\theta(p)$ from $p'$. Clearly, $\theta'$ replaces exactly one variable with a word containing the letter $b$, and this word contains exactly one $b$. 
      Moreover, $\omega'_1\cdots\omega'_{j+1}$ contains the letter $b$ exactly $k$ times. Therefore, the word
$$
u':=\theta'(\vec{x}'_1\omega'_1\vec{x}'_2\omega'_2\cdots\omega'_{j}\vec{x}'_{j+1}\omega'_{j+1})
$$
has either exactly $k$ or exactly $k+1$ occurrences of the letter $b$.
In the first case, in $\theta(p)$, the $k$th occurrence of $b$ is followed by $ab$, while in $\theta'(p')$, the $k$th occurrence of $b$ is followed by $aa$, since $n \geq 2$. In the second case, in $\theta(p)$, the $(k{+}1)$st occurrence of $b$ is preceded by $a$, while in $\theta'(p')$, the $(k{+}1)$st occurrence of $b$ is preceded by $b$, since $m_j \geq 2$. In each case, we reached a contradiction.
        Thus, we obtain $\theta(p)\in\mathcal{S}_{\varepsilon, 2}(p,R)\setminus \varepsilon L_{len}(p', R')$, yielding $(\star)$.

        \item[7.8.] $\omega'_{j+1} = b^{m_j} a^{n}$ and $\omega_{j+1} = a^n$ where $n = 1$ or $m_j = 1$. 
        
        Then $\omega_i \vec{x}_{i+1} \omega_{i+1} \ldots \vec{x}_{j+1} \omega_{j+1}$ and $\omega'_i \vec{x}'_{i+1} \omega'_{i+1} \ldots \vec{x}'_{j+1} \omega'_{j+1}$ are telltale conjugate subpatterns of $p$, resp.\ $p'$. This contradicts the premise of the Lemma.

         \item[7.9.] $\omega'_{j+1} = b^{m_j} a^{n}$ and $\omega_{j+1} = a^n b^m$  for some $n, m$ with $m, n \geq 2$.

         We use an argument similar to that used in Cases 6 and 7.7. Consider an $R'$-substitution $\theta'$ that replaces the first variable $x$ in $\vec{x}'_{j+2}$ with $ba$, all other variables in $[x]$ with $bb$, and erases all other variables.

Let $k$ be such that $\omega'_1 \cdots \omega'_{j+1}$ contains the letter $a$ exactly $k$ times. 
        Then, the word $\theta'(p')$ contains the substring $ba$ immediately after its $k$th occurrence of the letter $a$.

        Suppose $\theta$ is an $R$-substitution that generates $\theta'(p')$ from $p$. Then $\theta$ replaces exactly one variable with a word containing the letter $a$, and this word contains exactly one $a$. Moreover, $\omega_1\cdots\omega_{j+1}$ contains the letter $a$ exactly $k$ times. Therefore, the word
$$
u:=\theta(\vec{x}_1\omega_1\vec{x}_2\omega_2\cdots\omega_{j}\vec{x}_{j+1}\omega_{j+1})
$$
has either exactly $k$ or exactly $k+1$ occurrences of the letter $a$.

If $u$ has exactly $k+1$ occurrences of the letter $a$, then $n\ge 2$ implies that the $(k+1)$st occurrence of $a$ in $\theta(p)$ is immediately preceded by another $a$. By contrast, clearly, the $(k+1)$st occurrence of $a$ in $\theta'(p')$ is immediately preceded by the letter $b$, which contradicts $\theta(p)=\theta'(p')$. If $u$ has exactly $k$ occurrences of the letter $a$, then the $k$th occurrence of $a$ in $\theta(p)$ is immediately followed by $bb$, since $m\ge 2$. By contrast, the $k$th occurrence of $a$ in $\theta'(p')$ is immediately followed by $ba$, which again contradicts $\theta(p)=\theta'(p')$.
        Thus, we obtain $\theta'(p')\in \varepsilon L_{len}(p',R') \setminus \varepsilon L_{len}(p, R)$, yielding $(\star)$.

         \item[7.10.] $\omega'_{j+1} = b^{m_j} a^{n}$ and $\omega_{j+1} = a^n b^m$ for some $n, m$ with $m = 1$ or  $n = 1$. 
         
         This is in contradiction with the premise of Case 7. In particular, it contradicts $j$ being the maximal index such that $\omega_j = a^{n_j} b^{m_j}$ and $\omega'_j = b^{m_{j-1}} a^{n_j}$ where $n_j$, $m_j$, and $m_{j-1}$ are as in Definition \ref{def:conjugate-pat}. Hence, this case cannot occur.

         \item[7.11.] $\omega'_{j+1} = b^{m_j} a^{n}$ for some $n$ and $\omega_{j+1} \in a^n b^m a \Sigma^\ast$ for some $m\ge 1$.

         Consider an $R'$-substitution $\theta'$ that replaces one variable $x$ in $\vec{x}'_{j+2}$  (as well as all variables in $[x]$) with $b$ and erases all other variables.

        Let $k$ be such that the last occurrence of $a$ in $\omega'_{j+1}$ corresponds to the $k$th occurrence of the letter $a$ in $\theta'(p')$. 
        Obviously, there are more than $m$ occurrences of $b$ between the $k$th and $(k+1)$st occurrences of $a$ in $\theta'(p')$. On the other hand, since any $R$-substitution $\theta$ generating $\theta'(p')$ from $p$ must replace each variable in $p$ with a word from $\{ b \}^\ast$, the number of occurrences of $b$ between the $k$th and $(k+1)$st occurrences of $a$ in $\theta(p)$ remains $m$ (note $\omega_{j+1} \in a^n b^m a \Sigma^\ast$). Thus $\theta'(p')\in\varepsilon L_{len}(p',R')\setminus \varepsilon L_{len}(p,R)$ and $(\star)$ holds.
    \end{itemize}
    \end{proof}

\end{appendices}

\end{document}